\documentclass[conference]{IEEEtran}
\IEEEoverridecommandlockouts
\pdfoutput=1
\usepackage{cite}
\usepackage{amsmath,amssymb,amsfonts}
\usepackage{amsthm}
\usepackage{algorithm}
\usepackage[noend]{algpseudocode}
\usepackage{graphicx}
\usepackage{textcomp}
\usepackage{xcolor}
\usepackage{caption}
\usepackage{subcaption}
\usepackage{hyperref}
\usepackage{enumitem}
\usepackage{soul}
\usepackage{tikz}
\usepackage{booktabs}

\usetikzlibrary{shapes, positioning}
\usepackage[normalem]{ulem}
\usepackage{microtype}
\algdef{SE}[DOWHILE]{Do}{doWhile}{\algorithmicdo}[1]{\algorithmicwhile\ #1}%

\newtheorem{example}{Example}
\newtheorem{theorem}{Theorem}

\newtheorem{definition}{Definition}

\def\BibTeX{{\rm B\kern-.05em{\sc i\kern-.025em b}\kern-.08em
    T\kern-.1667em\lower.7ex\hbox{E}\kern-.125emX}}
\begin{document}

\title{HIGGS: \underline{HI}erarchy-\underline{G}uided \underline{G}raph Stream \underline{S}ummarization}

% \author{\IEEEauthorblockN{Xuan Zhao}
% \IEEEauthorblockA{University of Science and\\
% Technology of China}
% \and
% \IEEEauthorblockN{Xike Xie}
% \IEEEauthorblockA{University of Science and\\
% Technology of China}
% \and
% \IEEEauthorblockN{Christian S. Jensen}
% \IEEEauthorblockA{Department of Computer Science\\
% Aalborg University}}

\author{Xuan Zhao$^{1,3}$, Xike Xie$^{2,3}$, Christian S. Jensen$^{4}$ \\
	$^1$School of Computer Science, University of Science and Technology of China (USTC), China \\ $^2$School of Biomedical Engineering, USTC,  China\\
	$^3$Data Darkness Lab, MIRACLE Center, Suzhou Institute for Advanced Research, USTC, China\\ $^4$Aalborg University, Aalborg, Denmark\\
	\texttt{zhaoxuan2118@mail.ustc.edu.cn}, \texttt{xkxie@ustc.edu.cn}, \texttt{csj@cs.aau.dk}\\
}

\maketitle

\begin{abstract}
Graph stream summarization refers to the process of processing a continuous stream of edges that form a rapidly evolving graph.
The primary challenges in handling graph streams include the impracticality of fully storing the ever-growing datasets and the complexity of supporting graph queries that involve both topological and temporal information.
Recent advancements, such as PGSS and Horae, address these limitations by using domain-based, top-down multi-layer structures in the form of compressed matrices. However, they either suffer from poor query accuracy, incur substantial space overheads, or have low query efficiency.

This study proposes a novel item-based, bottom-up hierarchical structure, called {\it HIGGS}. Unlike existing approaches, HIGGS leverages its hierarchical structure to localize storage and query processing, thereby confining changes and hash conflicts to small and manageable subtrees, yielding notable performance improvements.
HIGGS offers tighter theoretical bounds on query accuracy and space cost.
Extensive empirical studies on real graph streams demonstrate that, compared to state-of-the-art methods, HIGGS is capable of notable performance enhancements: it can improve accuracy by over $3$ orders of magnitude, reduce space overhead by an average of $30$\%, increase throughput by more than $5$ times, and decrease query latency by nearly $2$ orders of magnitude.
\end{abstract}


\section{INTRODUCTION}

The importance of graph stream summarization has been underscored by applications of efficient processing and querying large volumes of dynamic graph data over specified temporal ranges~\cite{sarmento2016social,aggarwal2014evolutionary, long2011towards,cordeiro2016online,gou2019fast, chen2022horae,jiang2023auxo, ke2022multi, tian2021exploiting,ficel2021graph}. For example, in social network analysis, it helps to detect trending topics and the evolution of discussions over defined temporal intervals, enhancing insights into user engagement~\cite{long2011towards,cordeiro2016online,gou2019fast}. In pandemic analysis, it enables rapid identification of infection clusters and spread patterns during specific time periods~\cite{chen2022horae,jiang2023auxo}. Financial organizations use it to quickly identify fraudulent transaction patterns within certain time frames~\cite{ke2022multi,gaber2005mining}. Urban traffic management~\cite{tian2021exploiting,ficel2021graph,fang2021mdtp} benefits from analyzing and optimizing traffic flow based on historical data during peak hours or special events, specified by time ranges.

% \begin{figure}[t]
%     \centering
%     \includegraphics[width=\linewidth]{figContent/compareFrame.pdf}
%     \caption{Comparison of the two frameworks}
%     \label{fig:compareFrameworks}
%     \vspace{-15pt}
% \end{figure}
\begin{figure}[t]
    \centering
    \begin{subfigure}[b]{0.48\columnwidth}
        \includegraphics[width=\textwidth]{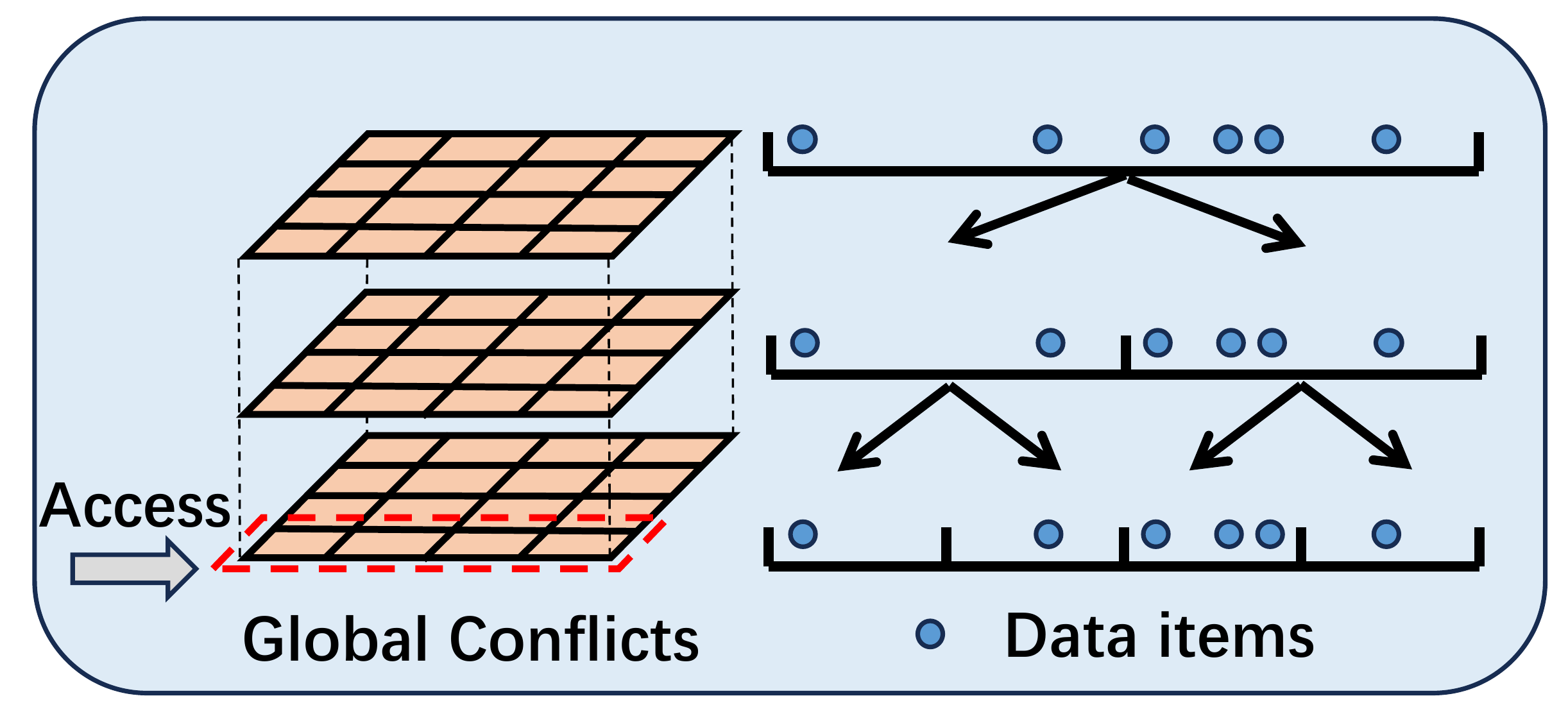}
        \caption{Top-down, Domain-based
        }
        \label{fig:preFrame}
    \end{subfigure}
    \hfill
    \begin{subfigure}[b]{0.48\columnwidth}
        \includegraphics[width=\textwidth]{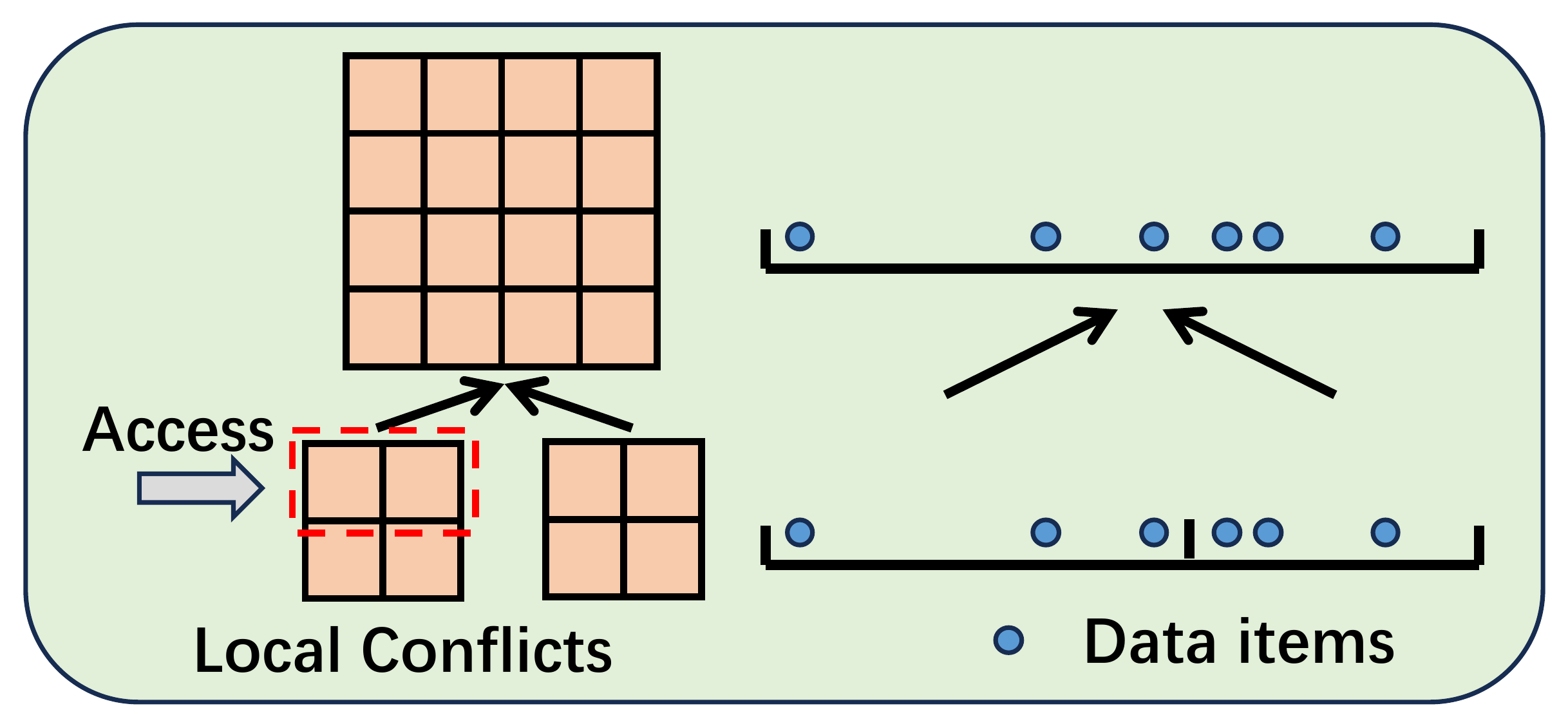}
        \caption{Bottom-up, Item-based}
        \label{fig:nowFrame}
    \end{subfigure}
    \vspace{-5pt}
    \caption{Architectures of Graph Stream Summarization}
    \label{fig:compareFrameworks}
    \vspace{-22pt}
\end{figure}

%Graph streams have become crucial in various applications that require real-time and dynamic graph analysis~\cite{sarmento2016social,aggarwal2014evolutionary}, network traffic monitoring \cite{so2009survey,d2019survey}, and personalized content and product recommendations \cite{tian2021exploiting,ficel2021graph}., \hl{such as social network analysis, network traffic monitoring, and personalized content and product recommendations.
%For example, social networking platforms like Twitter generate and update vast amounts of data through billions of daily interactions, including tweets, comments, likes, and follows. Similarly, large urban traffic management systems rely on real-time monitoring and analysis of numerous nodes (vehicles, traffic lights, surveillance cameras) and edges (vehicle movement paths, road connections) to adjust traffic lights and issue traffic guidance.}

Formally, a graph stream \cite{gou2021sliding,li2019time,pacaci2020regular,pacaci2022evaluating,aggarwal2010dense, mcgregor2014graph,suzumura2014towards, aggarwal2011outlier} represents an evolving graph $G = (V,E)$ as stream items arrive. Items in a graph stream are typically represented in the form $(s_{i},d_{i},w_{i},t_{i})$, corresponding to a directed edge from vertices $s_{i}$ to $d_{i}$ in $G$, where the edge has a weight of $w_{i}$ at arrival timestamp $t_{i}$.
In a graph stream, edges between the same source and destination vertices can appear multiple times with different weights and timestamps.
%As graph streams are massive in scale and highly dynamic, processing and storing this data in real-time, as well as performing complex queries on it, pose significant technical challenges.
Due to the immense scale and high dynamism of graph streams, real-time processing and storage of this data, along with executing complex queries, present significant technical challenges \cite{aggarwal2003framework,wu2006high}.

%However, in the era of big data, graph streams are often massive in scale, and characterized by high dynamism and real-time. For example, on the social networking platform Twitter \cite{myers2014information}, billions of tweets, comments, likes, and follow activities occur daily, continuously generating and updating graph stream data. In large urban traffic management systems \cite{nellore2016survey}, the system adjusts traffic light controls and issues traffic guidance by real-time monitoring and analysis of tens of thousands of nodes (such as vehicles, traffic lights, surveillance cameras) and edges (such as vehicle movement paths and road connections). Therefore, how to process and store graph stream data in real-time, as well as how to perform various complex queries on it, has become a challenging problem for researchers.

%Techniques of graph stream summarization \cite{tang2016graph, khan2016query,gou2019fast, jiang2023auxo,chen2022horae,jia2023persistent} tackle the problem by tolerating minimal precision loss to compress graph streams without temporal information, based on count-min sketches.
Pioneering works on graph stream summarization techniques \cite{tang2016graph, khan2016query,gou2019fast, jiang2023auxo,chen2022horae,jia2023persistent} address the problem by employing count-min sketches \cite{cormode2005improved} to compress graph streams, tolerating minimal precision loss and typically excluding temporal information.
The count-min sketch uses a zero-initialized 2D array of buckets and multiple hash functions to map stream items to array indices. Each hash function increments its counter on insertion, and querying an item returns the minimum value among the relevant counters.
%The count-min sketch employs multiple arrays and corresponding hash functions to store data streams, and it reports the minimum value among these hashed counters as the query result.
%thereby efficiently supporting various queries related to graph topology structures and attributes such as timestamps.
Tang et al. proposed TCM \cite{tang2016graph}, a sophisticated extension of the count-min sketch, to support graph streams. TCM uses a set of compressed matrices, each functioning like a count-min sketch, to efficiently sketch graph streams.
%TCM \cite{tang2016graph}, proposed by Tang et al., uses a set of adjacency matrices and hash functions to store and query graph stream data.
Gou et al. introduced GSS \cite{gou2019fast}, which extends TCM by combining an compressed matrix with an adjacency list and uses fingerprints as identifiers for nodes along with square hashing to enhance update and query performance.
Jiang et al. proposed Auxo \cite{jiang2023auxo}, which extends GSS by organizing the compressed matrices into a tree structure and setting the fingerprints as tree prefixes.
% which is built on the prefix embedded tree (PET). PET leverages binary logarithmic search and common binary prefixes embedding to provide efficiency and scalability.
% which utilizes the prefix embedded tree to organize compressed matrices.
To evaluate a vertex query~\cite{tang2016graph}, these methods access and aggregate all buckets storing the incident edges of the specified vertex, usually encompassing a row or column of the matrix.
Notably, these summarization techniques focus on non-temporal graph queries.
Extending these techniques directly to encompass multi-layer and multi-granularity aspects required in temporal information-aware graph summarization is non-trivial.

Recent advances in graph stream summarization have incorporated temporal information, targeted on supporting {\it \underline{T}emporal \underline{R}ange \underline{Q}ueries} (TRQ {\it in short}) \cite{jia2023persistent,chen2022horae}.
TRQ tasks aim to retrieve graph streams within a specified time interval, effectively answering graph queries, such as vertex, edge, path, and subgraph queries, concerning events that occurred within that time window.
Horae \cite{chen2022horae} and PGSS \cite{jia2023persistent} represent the latest advancements in TRQ-aided graph stream summarization, both based on TCM, with Horae achieving the state-of-the-art (SOTA) performance.
%%%%%%%%%%%%%%%%%%%%%% zhao del 1024 %%%%%%%%%%%%%%%%%%%%%%%%%%%%%%%%%
% Horae \cite{chen2022horae} extends TCM by incorporating multiple layers of counters into each bucket of the compressed matrix, where each layer aligns with a distinct time granularity established through a top-down domain-splitting strategy. For a specific time granularity or layer, the segmented counters across different buckets collectively constitute a TCM tailored for graph summarization at that layer.
Fig.~\ref{fig:preFrame} shows an example of Horae's architecture, demonstrating top-down recursive splitting of the temporal domain into $3$ layers of uniform temporal granularities. When evaluating a vertex query within a temporal range, the query range is decomposed into multiple sub-ranges corresponding to different temporal granularities.
Each sub-range necessitates a query on a specific layer. The results
from these layers are then aggregated to finalize the query result.
Notice that the previously mentioned multi-layer structure is not hierarchical, meaning there is no explicit parent-child relationship between layers.
%%%%%%%%%%%%%%%%%%%%% zhao del 1024 %%%%%%%%%%%%%%%%%%%%%%%%%%%%%%%%%%%%%%
% Instead of being organized hierarchically, the temporal information is encoded within the entries stored in the matrix's buckets, identified by binary strings known as fingerprints.
Then, each layer's compressed matrix is constructed using the entire data stream on a global scale.

There are several drawbacks in this non-hierarchical
design that have hindered the performance of graph stream summarization, in terms of query accuracy, time efficiency, and space efficiency.
{\bf 1)} \underline{\it Global Hashing Conflicts}: 
Decomposing a query into different layers leads to global hashing conflicts at each layer,
%When a query is decomposed
%into different layers, each layer experiences global hashing conflicts, 
which deteriorate query accuracy. 
{\bf 2)} \underline{\it Lack of Temporal Hierarchy}:
The absence of an explicit hierarchical structure for
organizing temporal information impacts insertion and query efficiency. For example, evaluating a vertex query generates a equivalent class of sub-range queries, each requiring access to the matrix built on the global data stream and resulting in significant runtime overhead.
{\bf 3)} \underline{\it Inadaptivity to Stream Irregularity}: The flat, uniform grid-like structure is not adaptive to the irregularity of graph streams.
For example, there are only $2$ leaves and $2$ layers in Fig.~\ref{fig:nowFrame} (bottom-up, item-based), in comparison to $4$ leaves and $3$ layers in Fig.~\ref{fig:preFrame} (top-down, domain-based).
%{\color{blue}hence its tree height is taller.}
Additionally, its reliance on predefined matrix sizes undermines its adaptivity to the stream irregularity.

The irregularity of graph streams stems from two factors. {\bf 1)}
\underline{\it Skewed Vertex Degree Distribution}: Vertex degrees follow a
power-law distribution, implying a few high-degree vertices (head
vertices) connect to a large number of low-degree vertices,
as shown in Fig.~\ref{fig:frequencyDegree}.
{\bf 2)} \underline{\it Impact of Head Vertices}: Changes
in a head vertex lead to significant alterations in its numerous incident edges, causing substantial variations in graph streams. In Fig.~\ref{fig:frequencyEdge}, it shows the distribution of hot time intervals where a
large number stream edges occur. The irregularity plays a crucial role in balancing space cost with query accuracy and efficiency, which is the essence of graph stream summarization.

\begin{figure}[t]
    \centering
    \begin{subfigure}[b]{0.32\columnwidth}
        \includegraphics[width=\textwidth]{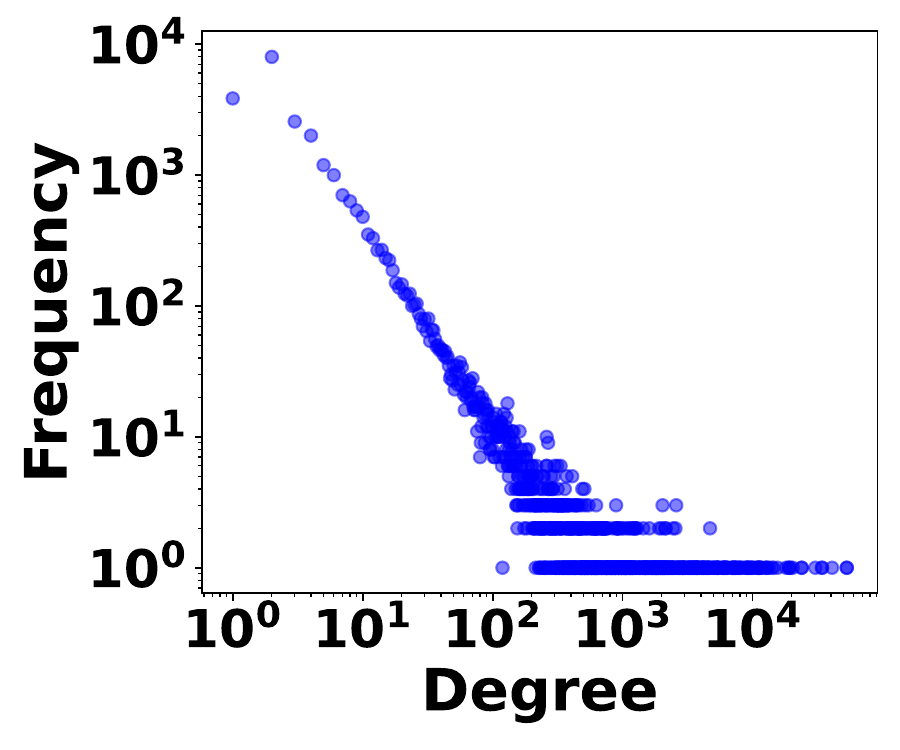}
       % \vspace{-13pt}
        \caption{Lkml}
        \label{fig:skewnessLkml2}
    \end{subfigure}
    \hfill
    \begin{subfigure}[b]{0.32\columnwidth}
        \includegraphics[width=\textwidth]{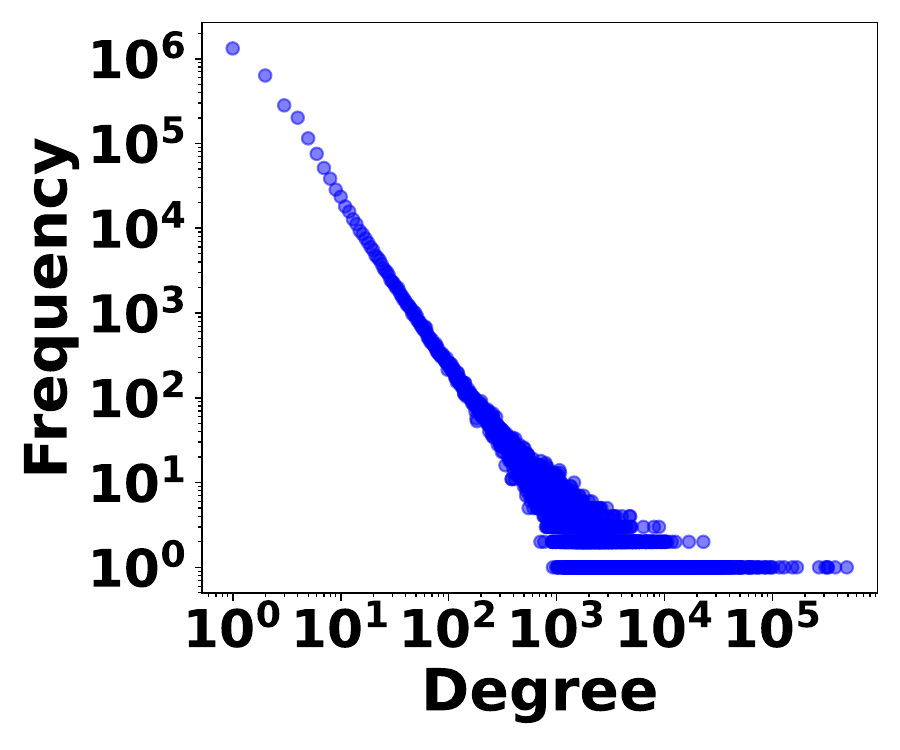}
      %  \vspace{-13pt}
        \caption{Wiki-talk}
        \label{fig:skewnessWiki2}
    \end{subfigure}
    \hfill
    \begin{subfigure}[b]{0.32\columnwidth}
        \includegraphics[width=\textwidth]{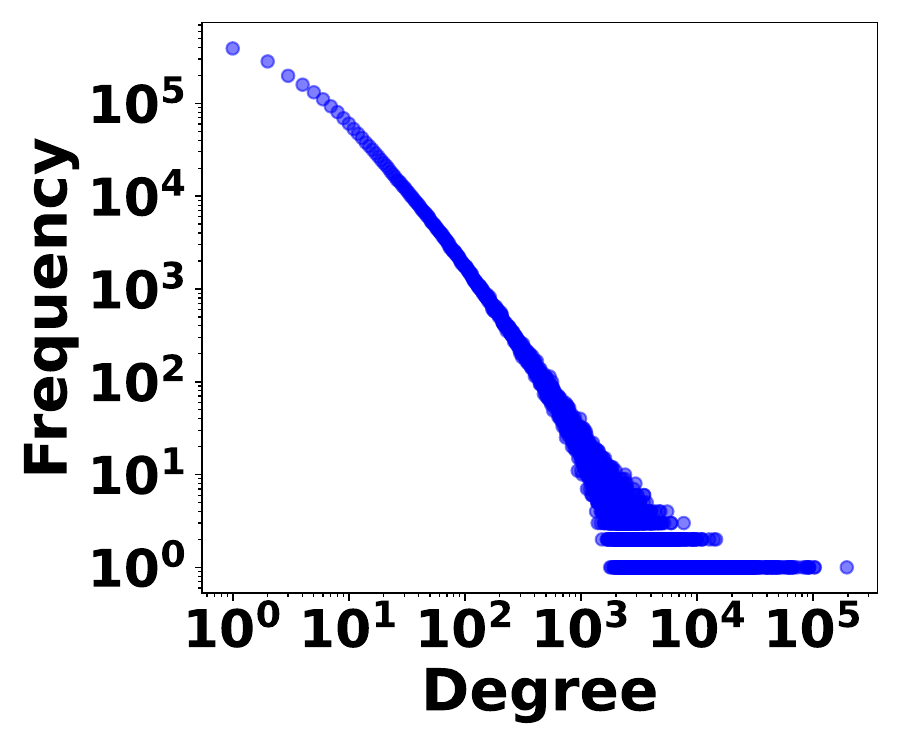}
     %   \vspace{-13pt}
        \caption{Stackoverflow}
        \label{fig:skewnessStack2}
    \end{subfigure}
    \vspace{-5pt}
    \caption{Skewness of Vertex Degrees}
    \label{fig:frequencyDegree}
%\end{figure}
%
%\begin{figure}[htb]
%    \centering
%        \begin{subfigure}[b]{1\columnwidth}
%        \includegraphics[width=0.32\textwidth]{figContent/skewnessLkml1.pdf}
%        \includegraphics[width=0.32\textwidth]{figContent/skewnessWiki1.pdf}
%        \includegraphics[width=0.32\textwidth]{figContent/skewnessStack1.pdf}
%        \caption{Stackoverflow}
%        \label{fig:skewnessStack1}
%    \end{subfigure}
    \begin{subfigure}[b]{0.32\columnwidth}
        \includegraphics[width=\textwidth]{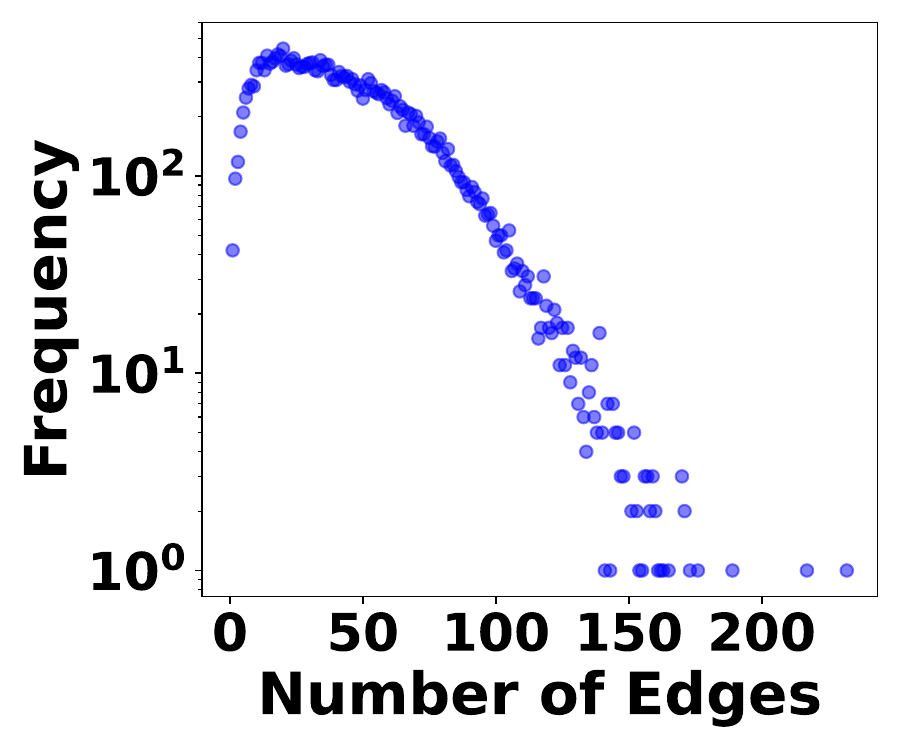}
                %\vspace{-13pt}
        \caption{Lkml}
        \label{fig:skewnessLkml1}
    \end{subfigure}
    \hfill
    \begin{subfigure}[b]{0.32\columnwidth}
        \includegraphics[width=\textwidth]{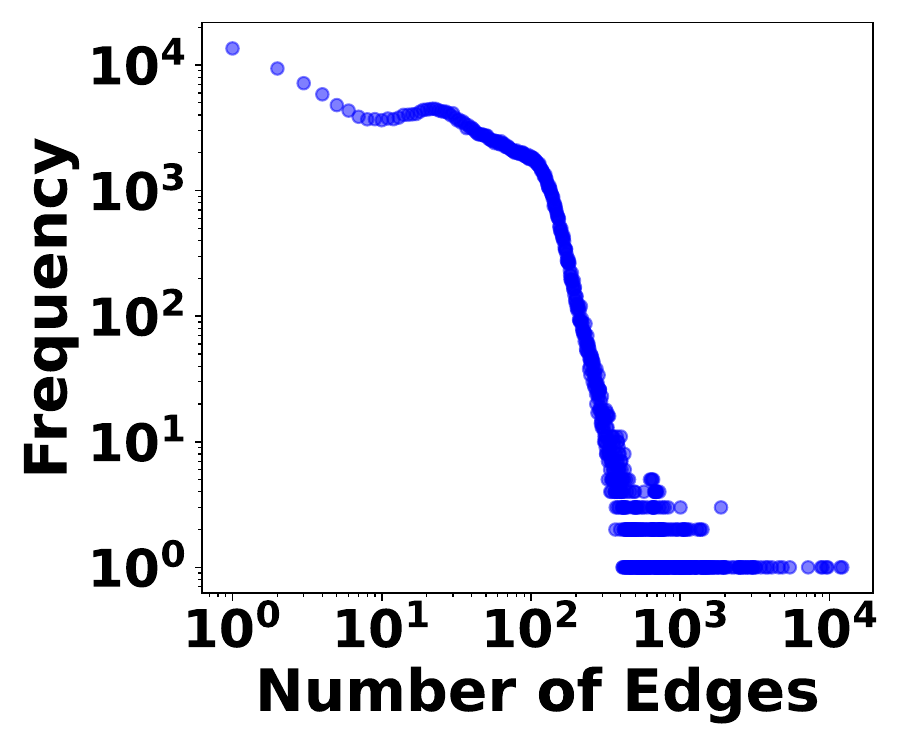}
                %\vspace{-13pt}
        \caption{Wiki-talk}
        \label{fig:skewnessWiki1}
    \end{subfigure}
    \hfill
    \begin{subfigure}[b]{0.32\columnwidth}
        \includegraphics[width=\textwidth]{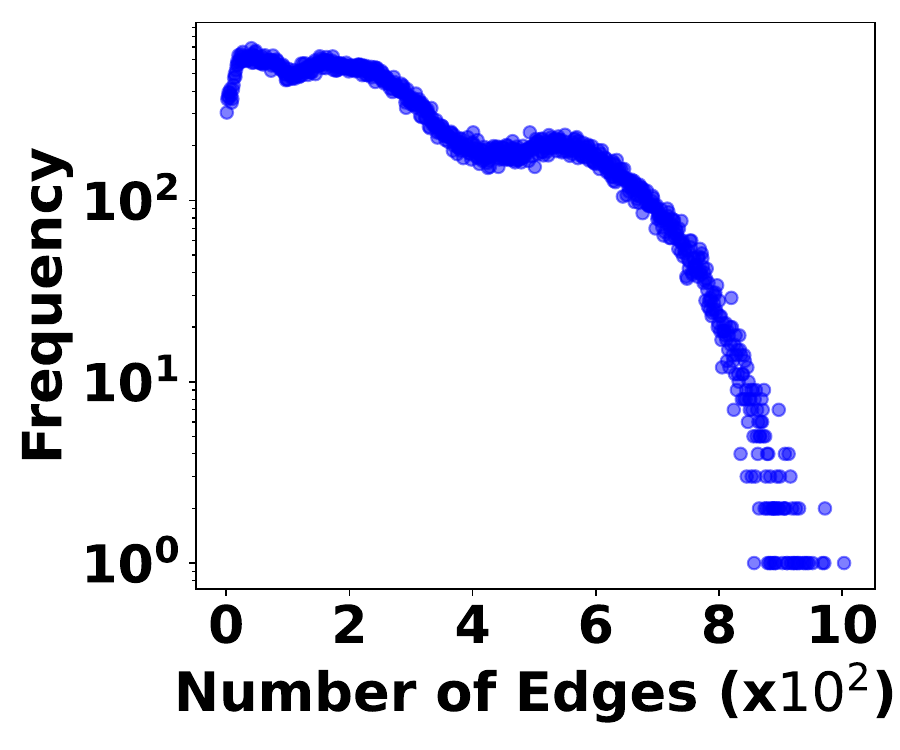}
                %\vspace{-13pt}
        \caption{Stackoverflow}
        \label{fig:skewnessStack1}
    \end{subfigure}
    \vspace{-5pt}
    \caption{Irregularity of Graph Stream Item Arrivals}
    \vspace{-15pt}
    \label{fig:frequencyEdge}
\end{figure}

We propose HIGGS, a novel item-based,
bottom-up hierarchical data structure designed for summarizing
graph streams with temporal information. HIGGS functions
akin to an aggregated B-tree, as shown in Fig.~\ref{fig:nowFrame}, where each tree
node corresponds to a specific time interval and integrates a compressed
matrix representing summarized graph data from its subtree.
On the bottom layer, it adaptively distributes irregular stream
edges to storage buckets, optimizing space utilization and reducing the tree height. Unlike existing approaches relying on global structural information, HIGGS utilizes its hierarchical structure to
localize storage and query processing and confine the changes and hash conflicts to small and managable subtrees. Moreover, HIGGS employs a meticulously designed expansion mechanism for compressed matrices, along with data aggregation and query evaluation, which not only enhance query efficiency but also ensure that no additional error is introduced during the aggregation process. As a result, HIGGS achieves significant improvements in query accuracy, as well as time and space efficiency.
%%%%%%%%%%%%%% zhao del 1024 %%%%%%%%%%%%%%%%%%%%%%%%%%%%%%%%%%%%%
% This approach enhances query accuracy, space efficiency, and processing speed.
Additionally, we have implemented various optimizations to further enhance the performance
and conducted detailed theoretical analysis to
verify its advantages over existing works.
Comprehensive
and thorough experiments further substantiate our proposal, with
results indicating that HIGGS surpasses other methods in accuracy by at least three orders of magnitude. In summary, the major contributions of this work are as follows:

%{\color{blue}
\begin{itemize}[leftmargin=5pt]
    \item We propose HIGGS, a novel item-based, bottom-up hierarchical structure designed for summarizing graph streams with temporal information. In addition, we design a expansion mechanism for compressed matrices, as well as data aggregation and query algorithms.
    \item This structure leverages its hierarchical organization to localize storage and query processing, effectively confining changes and hash conflicts to small, manageable subtrees. By integrating data aggregation and query algorithms, it achieves significant performance improvements.
    %%%%%%%%%%%%%%%%%% zhao del 1024 %%%%%%%%%%%%%%%%%%%%%%%%%%%%%%%%
    % thereby achieving significant performance improvements.
    \item We provide a detailed theoretical analysis of our proposal, focusing on space and time efficiency, and most importantly, query accuracy, guaranteed by tighter theoretical bounds.
    \item We report on comprehensive and thorough experiments, finding that HIGGS is capable of notable performance enhancements: it can improve accuracy by over 3 orders of magnitude, reduce space overhead by an average of 30\%, increase throughput by more than 5 times, and decrease query latency by nearly 2 orders of magnitude.
\end{itemize}

\begin{table}[ht]
    \centering
    \vspace{-7pt}
    \caption{Frequently used Notation}
    \vspace{-5pt}
    \label{tab:symbol}
    \begin{tabular}{cl}
        \hline
        Symbol      & Description                                                                     \\
        \hline
        $|V|$ and $|E|$       & the number of vertices and edges in the graph stream \\
        %$|E|$       & the number of edges in the graph stream \\
        $d_{i}$     & the size of the compressed matrix of the $i$-{th} layer                            \\
        $F_{i}$     & the length of the $i$-{th} layer's fingerprint                                      \\
        $M^{i}_{k}$ & the $k$-{th} matrix in the $i$-{th} layer                                       \\
        $\theta$    & the maximum number of child nodes \\
        $R$         & the number of reduced fingerprint bits in aggregation \\
        $n_{i}$     & the number of nodes in the $i$-{th} layer \\
        $f(v)$      & the fingerprint of $v$ \\
        $h(v)$      & the address of $v$ \\
        $L$         & the length of the graph stream duration \\
        $L_{q}$     & the length of the time range to be queried \\
        $b$         & the number of entries in a bucket \\
        $Z$         & the size of the value range of the hash function \\
        \hline
    \end{tabular}
    \vspace{-10pt}
\end{table} 

\section{RELATED WORK}
{\bf Stream Summarization.} In the era of big data, managing and analyzing massive streams of data presents significant challenges. Current solutions often rely on sketches, a set of algorithms aimed at representing large datasets with concise summaries. Extensive literature covers various sketching techniques, including Count Sketch \cite{charikar2002finding}, count-min (CM) sketch \cite{cormode2005improved}, CU sketch \cite{estan2002new}, among others \cite{cohen2003spectral,deng2007new,yang2017pyramid,zhang2020off,liu2021xy,cao2023meta,liu2024probabilistic,gao2022multi,xie2016elite}. For example, the CM sketch utilizes $d$ arrays of counters, each associated with a unique hash function. Upon the arrival of a stream item, each hashed counter is incremented. Deletion reverses this process by decrementing the counters. During querying, the minimum value across the hashed counters is returned for the queried item. 

\begin{figure}[ht]
\vspace{-12pt}
    \centering
    \includegraphics[width=.9\linewidth]{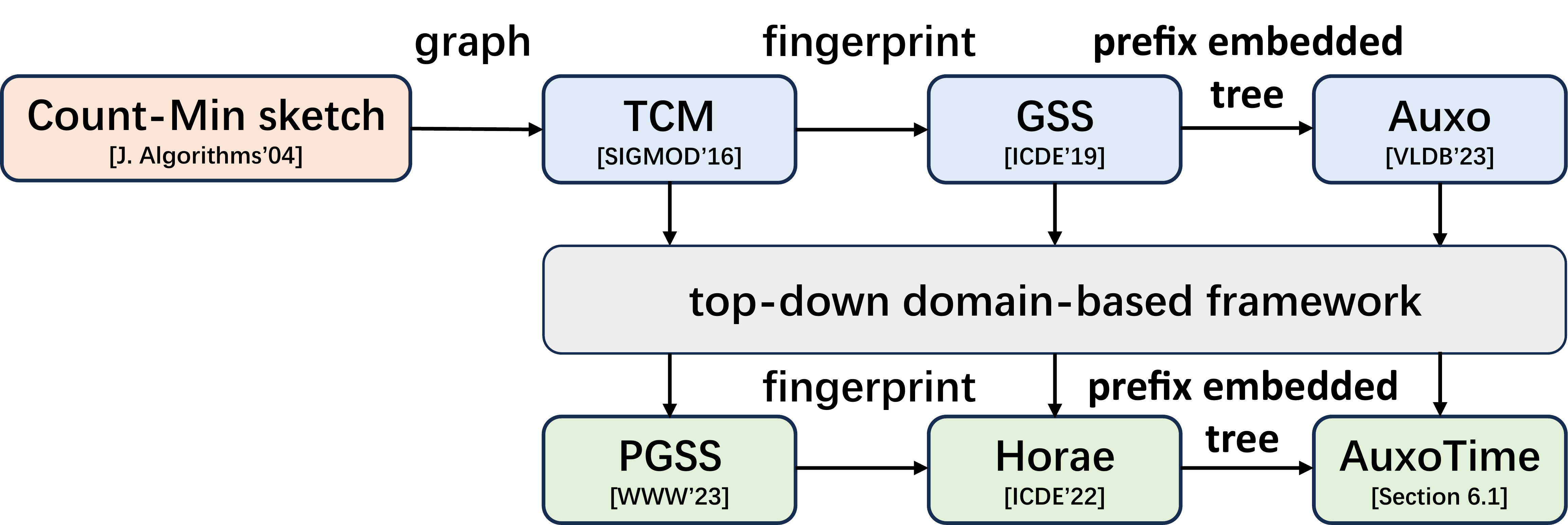}
    %\vspace{-13pt}
    \caption{Roadmap of Technical Evolution for Graph Stream Summarization (\cite{cormode2005improved}\cite{tang2016graph}\cite{gou2019fast}\cite{jiang2023auxo}\cite{jia2023persistent}\cite{chen2022horae})}
    \label{fig:Development}
    \vspace{-5pt}
\end{figure}

{\bf Graph Stream Summarization.}
Graph stream processing is inherently more complex than conventional data stream processing, due to the underlying graphical topology.
As a result, considerable research has focused on graph stream summarization \cite{tang2016graph,khan2016query,gou2019fast, jiang2023auxo,chen2022horae,jia2023persistent,ashrafi2021gs4,zhao2011gsketch,tsalouchidou2018scalable,hassan2018sbg,feng2023mayfly}. Fig.~\ref{fig:Development} shows the technical evolution of graph stream summarization. TCM, introduced by Tang et al. \cite{tang2016graph}, employs multiple matrices, each associated with a hash function. During updates, these hash functions determine the insertion locations based on the hash values of the edge's source and destination nodes, mapped to row and column addresses in the matrices. When querying, TCM retrieves the minimum aggregated weight from corresponding positions across all matrices. Despite its versatility in supporting vertex and path queries, TCM's accuracy suffers from significant hash collisions. Addressing this issue, Khan et al. proposed gMatrix \cite{khan2016query}, a variant of TCM using reversible hash functions to enhance query support but introducing additional errors.
To improve the accuracy of queries in TCM, Gou et al. proposed GSS \cite{gou2019fast}, which uses a matrix and an adjacency list for data storage, employing fingerprints for edge identification. When an edge is hashed into a matrix bucket, it checks whether the fingerprints match: if yes, the weights are aggregated; otherwise, the edge is inserted into the adjacent list. To reduce the size of the list, square hashing is employed to increase edge mapping positions in the matrix, thereby reducing insertions into the list.
However, handling large-scale graph streams with GSS remains challenging because the potential for high volumes can lead to longer lists, which in turn can degrade performance.
Jiang et al. introduced Auxo \cite{jiang2023auxo}, a scalable graph stream summarization method using the prefix embedded tree (PET). PET arranges compressed matrices in a tree structure and directs edge insertion based on prefix fingerprints. They also developed a proportional incremental strategy to improve PET’s memory efficiency. Consequently, Auxo exhibits good scalability with large-scale graph streams. However, like other existing methods mentioned, Auxo lacks support for temporal range queries.
%Additionally, they proposed a proportionally incremental strategy to improve the memory utilization of PET. Therefore, Auxo exhibits stronger scalability and better performance when facing large-scale graph streams, making it the state-of-the-art. However, like the aforementioned graph stream summarizations, it does not support time range queries.

{\bf Stream Summarization for TRQ.}
Recent advancements in supporting time range queries include many methods~\cite{wei2015persistent,peng2018persistent,shi2021time,fan2023hoppingsketch}.
Wei et al. introduced persistent CountMin and AMS sketches \cite{wei2015persistent}, which enable querying historical data streams. Peng et al. proposed the Persistent Bloom Filter \cite{peng2018persistent}, designed to identify elements within specific time intervals. Shi et al. developed at-the-time persistent (ATTP) and back-in-time persistent (BITP) sketches \cite{shi2021time} for temporal analytics on data streams. Fan et al. enhanced temporal membership queries with the hopping sketch \cite{fan2023hoppingsketch}, using a sliding window approach for accurate frequency queries. However, these methods do not cater specifically to graph streams and do not preserve graph topological structures, limiting their ability to execute queries based on graph topology.
%Recent works supporting time range queries include \cite{wei2015persistent,peng2018persistent,shi2021time,fan2022hoppingsketch}, where Wei et al. designed persistent sketches \cite{wei2015persistent}, including persistent CountMin sketch and persistent AMS sketch, that can answer stream-based queries at any previous time. Peng et al. proposed the Persistent Bloom Filter \cite{peng2018persistent}, which can determine if an element appeared within a specific time range. Shi et al. built two models for temporal analytics on data streams: at-the-time persistent (ATTP) and back-in-time persistent (BITP) sketches \cite{shi2021time}. Fan et al. improved on the PBF by proposing the hopping sketch \cite{fan2022hoppingsketch}, which utilizes a sliding window to support temporal membership queries and also allows for accurate frequency queries. However, these works are not targeted at graph streams and do not maintain the topological structure of graphs, which means they cannot perform queries based on graph structure.

{\bf Graph Stream Summarization for TRQ.}
Cutting-edge works that can maintain graphical topology and support time range queries include PGSS \cite{jia2023persistent} by Jia et al. and Horae \cite{chen2022horae} by Chen et al. PGSS extends TCM by storing multiple arrays of counters in each matrix bucket, each corresponding to a different time granularity.
When inserting edges, it uses hash functions to locate bucket positions and updates counters based on the time granularity or layer. Horae, on the other hand, employs time prefix encoding in a multi-layer %graph stream summarization
framework. Each layer operates as a GSS with specific time granularity, encoding edges upon arrival and updating nodes and time prefixes accordingly. Both PGSS and Horae adopt a top-down, temporal domain-based multi-layer approach. During queries, they decompose time ranges into sub-ranges per layer granularity, query each sub-range, and aggregate results. %PGSS-MDC, rooted in TCM, exhibits lower accuracy compared to Horae, which leverages GSS for improved performance.
However, this top-down domain-based structure struggles with irregular graph streams, resulting in poor space efficiency. Each layer employs a compressed matrix for the entire stream, leading to potential hashing conflicts and reduced accuracy. Additionally, larger matrices per layer and the need to access more buckets for time-range queries lower query efficiency.

\section{PROBLEM STATEMENT}
%\textbf{Graph Stream}:

\begin{definition}[Graph Stream]
%\textbf{Graph Stream.}
A graph stream is a sequence of edges $S = \langle e_{1},e_{2},...,e_{n}\rangle$ arriving sequentially over time, where $e_{i} = (s_{i}, d_{i}, w_{i}, t_{i})$ represents a directed edge from vertices $s_{i}$ to $d_{i}$ at time $t_{i}$ with weight $w_{i}$.
\end{definition}
%The changes in the graph stream form a dynamic graph $G(V,E)$, where $V$ and $E$ respectively represent the sets of vertices and edges. Unlike static graphs, graph streams emphasize the dynamic nature and temporal sensitivity of graph data.
%Understanding the characteristics of graph streams within specified time ranges and efficiently querying these characteristics are of paramount importance.
%To support various queries associated with timestamp attributes, we initially define two time-range query primitives, from which other queries can be extended.
\vspace{-5pt}
Understanding graphical information for any temporal range is crucial, encompassing edge, vertex, path, and subgraph queries. These queries, referred to as {\it Temporal Range Queries} (TRQ in {\it short}), are essential for extracting meaningful insights and comprehending the evolving structure of the graph.
Two fundamental primitives for TRQ are {\it edge} and {\it vertex queries}:
%which can be formalized as:
%The two queries are essential for extracting meaningful insights and understanding the evolving structure of the graph.

\begin{definition}[TRQ Primitives]
%\textbf{Temporal-Range Query Primitives}:
Given a streaming graph $G(V,E)$ and a temporal range $I = [t_{s}, t_{e}]$, the two query primitives, node and edge queries can be defined as:
\begin{itemize}
    \item \textit{\bf Edge Queries.} Given a directed edge $e$ and a specific temporal range $I$, the query returns the aggregated weight of this edge within $I$.
    \item \textit{\bf Vertex Queries.} Given a vertex $v$ and a specific temporal range $I$, the query returns the aggregated weight of all outgoing (or incoming) edges within $I$.
\end{itemize}
\end{definition}

\vspace{-5pt}
%Utilizing these query primitives, we can reconstruct subgraphs of the original graph stream for any given time range, maintaining the topological structure. Thus, complex graph-structure-based queries are naturally supported.
Based on the two primitives, more advanced queries, such as path and subgraph queries, can be performed to gain deeper dynamical and graphical insights.
For example, in a subgraph query, we can perform an edge query for each edge within the subgraph, aggregate the weights, and return the results.

\begin{figure}[htbp]
\vspace{-10pt}
    \centering
    \includegraphics[width=\linewidth]{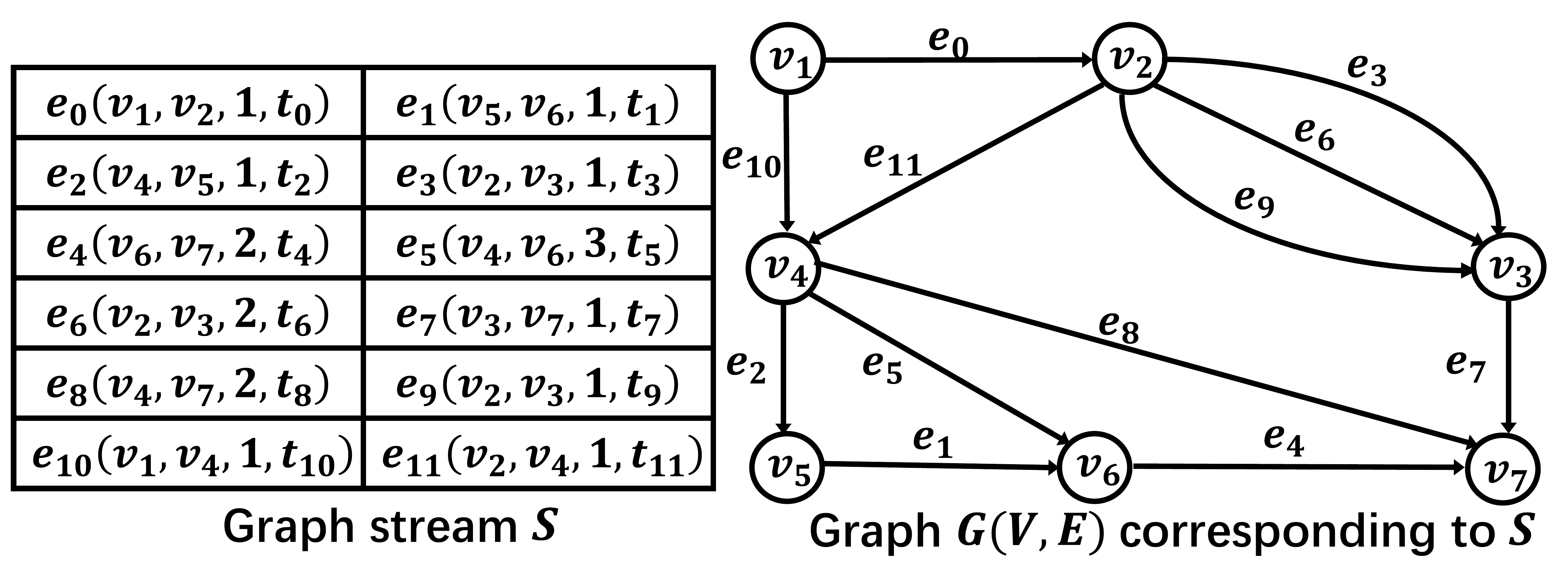}
\vspace{-16pt}
    \caption{An Example of TRQ Primitives}
\vspace{-14pt}
    \label{fig:graph stream}
\end{figure}

\begin{example}
Fig.~\ref{fig:graph stream} depicts the graph stream $S$ and its associated graph $G = (V, E)$. The aggregated weight of the directed edge $v_{2} \rightarrow v_{3}$ from $t_{5}$ to $t_{10}$ is $3$, the sum of weights at $t_{6}$ and $t_{9}$. For vertex queries, the total weight of $v_{4}$'s outgoing edges from $t_{1}$ to $t_{11}$ is $6$, calculated from the edges $(v_{4}, v_{5}, t_{2}, 1)$, $(v_{4}, v_{6}, t_{5}, 3)$, and $(v_{4}, v_{7}, t_{8}, 2)$. For the subgraph $\{(v_{2}, v_{3}), (v_{3}, v_{7}), (v_{2}, v_{4})\}$ between $t_{4}$ and $t_{8}$, only edges $v_{2} \rightarrow v_{3}$ and $v_{3} \rightarrow v_{7}$ contribute to the total weight of $3$, obtained at $t_{6}$ and $t_{7}$, respectively.
\end{example}

The symbols used in this paper are listed in Table~\ref{tab:symbol}.

\section{Methodology}
In this section, we report technical details about the proposed item-based, bottom-up hierarchical architecture. Section~\ref{subsec:frame} illustrates the overview of the architecture. Section~\ref{subsec:operation} investigates relevant operations, including construction, insertion, and query evaluation. Section~\ref{subsec:opt} discusses the optimization techniques for further enhancing the performance.
% To address a series of issues brought by the framework of Horae, we propose an efficient and flexible temporal graph stream summarization technique called HIGGS, along with a unique temporal decomposition and query scheme. Table 1 lists the mathematical symbols used in this paper.(note: remember to reorganize the contribution)

%To address a series of challenges inherent in existing frameworks, we have introduced a novel approach characterized by object-based, bottom-up hierarchical structures. This method synergizes compressed matrices and the principle of B-trees, culminating in the creation of HIGGS—a highly efficient, precise, and compact graph stream summarization structure. Furthermore, we propose the boundary search algorithm, which utilizes HIGGS to resolve graph queries that incorporate both topological and temporal information.

% Based on the analysis above, to address a series of issues brought by the existing framework, we propose an efficient, accurate, and compact graph stream summarization structure called HIGGS. It utilizes a novel, dynamically expandable multi-layer adaptive framework and our carefully designed Storage Block Granularity-Based Decomposition (SBGD) algorithm. Additionally, we have introduced a Time Range Search Tree (TRST) to organize the storage blocks in each layer. Table 1 lists the mathematical symbols used in this paper.
\subsection{Architecture}
\label{subsec:frame}

\begin{figure}[htbp]
    \centering
    \vspace{-10pt}
\includegraphics[width=\linewidth]{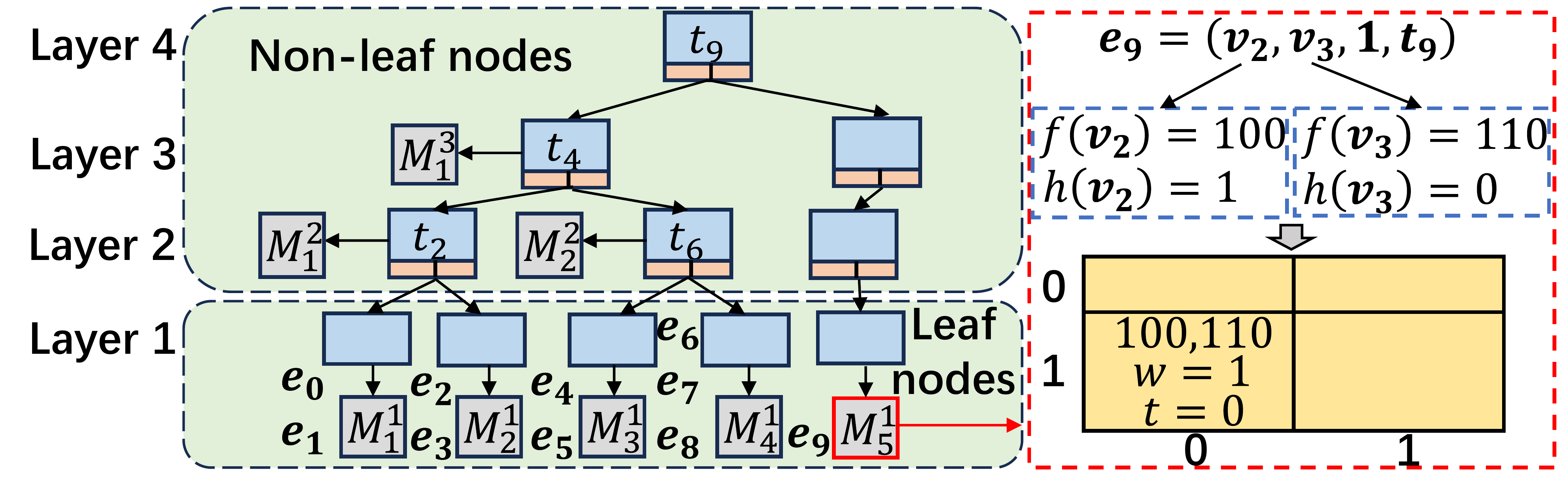}
    \vspace{-16pt}
    \caption{Overview of the Architecture of HIGGS}
    \vspace{-8pt}
    \label{fig:newHIGGS}
\end{figure}

The architecture of HIGGS is essentially an aggregated B-tree with the following properties:

\begin{itemize}%[left=0pt]
    \item Each tree node corresponds to a time interval and is with at most $\theta$ children, along with a compressed matrix, summarizing the graph stream for that interval.
    \item A non-leaf node with $k$ children contains $k-1$ keys, which are represented by timestamps acting as separation values dividing its subtrees. All leaves appear on the same layer, i.e., the bottom layer.
    \item Each bucket, essentially an element in the compressed matrix, contains a set of entries, each representing the tuple $(f(s_{i}), f(d_{i}), t_{i}, w_{i})$ for an edge $e_{i}$.
    \begin{itemize}
        \item $f(s_{i})$ and $f(d_{i})$ are fingerprints (compact identifiers) of the source and destination vertices, respectively;
        \item $t_{i}$ is the offset relative to the matrix's start time;
        \item $w_{i}$ is the weight of $e_{i}$.
    \end{itemize}
        %\item Each non-leaf node contains $\theta - 1$ keys and $\theta$ child nodes, along with a time interval and an associated compressed matrix, whereas each leaf node comprises a time interval and its associated compressed matrix.
    %\item Each key corresponds to a timestamp, enabling the indexing of these matrices.
    %\item Each bucket within the matrix contains multiple entries, with each entry representing the tuple $\langle f(s_{i}), f(d_{i}), t_{i}, w_{i} \rangle$ for an edge $e_{i}$.
    %\item A fingerprint acts as a compact identifier for nodes, using significantly less space, with detailed computational specifics provided in Section 4.4.
\end{itemize}

%This structure ensures that HIGGS effectively manages graph stream data while optimizing space usage and query performance.
%
%
%The framework of the HIGGS is essentially an aggregated B-tree, where each leaf node holds $\theta - 1$ keys and $\theta$ pointers to compressed matrices. Each key corresponds to a timestamp, enabling the indexing of these matrices. The primary purpose of the compressed matrix is to store graph stream data. Each bucket within the matrix contains multiple entries, with each entry representing the tuple $\langle f(s_{i}), f(d_{i}), t_{i}, w_{i}\rangle$ for an edge $e_{i}$. Here, $f(s_{i})$ and $f(d_{i})$ are the fingerprints of the source and destination nodes, respectively; $t_{i}$ is a timestamp relative to the matrix's start time; and $w_{i}$ is the weight of $e_{i}$. A fingerprint acts as a compact identifier for nodes, using significantly less space, with detailed computational specifics provided in Section 4.4.

%\textit{Example 3}:
\begin{example}
Fig.\ref{fig:newHIGGS} shows the HIGGS structure, which includes four layers: the bottom layer contains leaf nodes, and the top three layers comprise non-leaf nodes. Each node holds up to one key and can link to two children. Matrices, labeled as $M_{q}^{p}$ where $p$ is the level and $q$ the matrix number, are indexed by timestamp intervals $\{(t_{i}, t_{j})\}$, each covering a specific time range. For example, $M_{2}^{1}$ spans $t_{2}$ to $t_{4}$, while $M_{2}^{2}$ covers $t_{4}$ to $t_{9}$.
\end{example}

\vspace{-5pt}

The non-leaf nodes of HIGGS, like leaf nodes, store multiple keys and pointers to compressed matrices. These keys represent the aggregated temporal scope of their child nodes. Each non-leaf node matrix summarizes data from its child nodes, devoid of timestamps, containing entries as $(f(s_{i}), f(d_{i}), w_{i})$.

\vspace{-5pt}
\subsection{Operations of HIGGS}
\label{subsec:operation}
%Next, we shall delineate the intricate operational procedures of insertion, aggregation, and querying within the HIGGS framework.

%{\bf Construction.} The construction of HIGGS is object-based and bottom-up, optimally utilizing the sequence in which graph streams arrive. When an edge arrives, it is first inserted into the newest compression matrix at the first level, with a time complexity of $O(1)$. \hl{Should this insertion fail, a new matrix is constructed to store the edge, and its timestamp is transmitted to the parent node to delineate between matrices. If the parent node is a non-leaf node, a new matrix must be constructed to aggregate the data from all compression matrices pointed to by this child node, and the parent node must be updated to point to this new matrix.} Moreover, if the number of keys in the parent node has reached its maximum capacity, unlike traditional B-tree split operations, we will continue to transmit the timestamp upwards until successful insertion is achieved. If the insertion still fails upon reaching the root node, a new node must be created to store this timestamp, designating the original root as its first child node. It is important to note that when the insertion point of the timestamp is not a leaf node, to ensure that all matrices at the first level remain on the same layer, we employ filler nodes to facilitate connections.

{\bf Construction.} The construction of HIGGS is item-based and bottom-up, utilizing the sequence in which graph streams arrive. When an edge arrives, it is inserted into the newest compression matrix at the first level in $O(1)$ time. If the matrix is full, a new one is constructed to store the edge, and its timestamp is propagated upwards. If the root becomes full, a new root is created with the original one as its child. Filler nodes ensure that matrices connected to leaf nodes remain at the leaf layer, keeping the tree balanced.

%%%%%%%%%%%%%%%% zhao del 1030 %%%%%%%%%%%%%%%%%%%%%%%%%%%%%%%%%%%%
% The construction of the HIGGS is item-based and bottom-up, utilizing the sequence in which graph streams arrive. When an edge arrives, it is first inserted into the newest compression matrix at the first level, with a time complexity of $O(1)$. If the compression matrix is fully packed, a new matrix is constructed to store the edge, and its timestamp is propagated to the parent node. If the parent node is also fully packed, the process continues upward until successful insertion. If this process reaches the root node and the root is fully packed, a new root node is created, designating the original root as its first child node. Filler nodes are used to ensure that all matrices connected to leaf nodes are kept on the leaf layer, maintaining the balance of the entire tree structure.

\begin{figure}[htbp]
    \vspace{-8pt}
    \centering
    \includegraphics[width=0.8\linewidth]{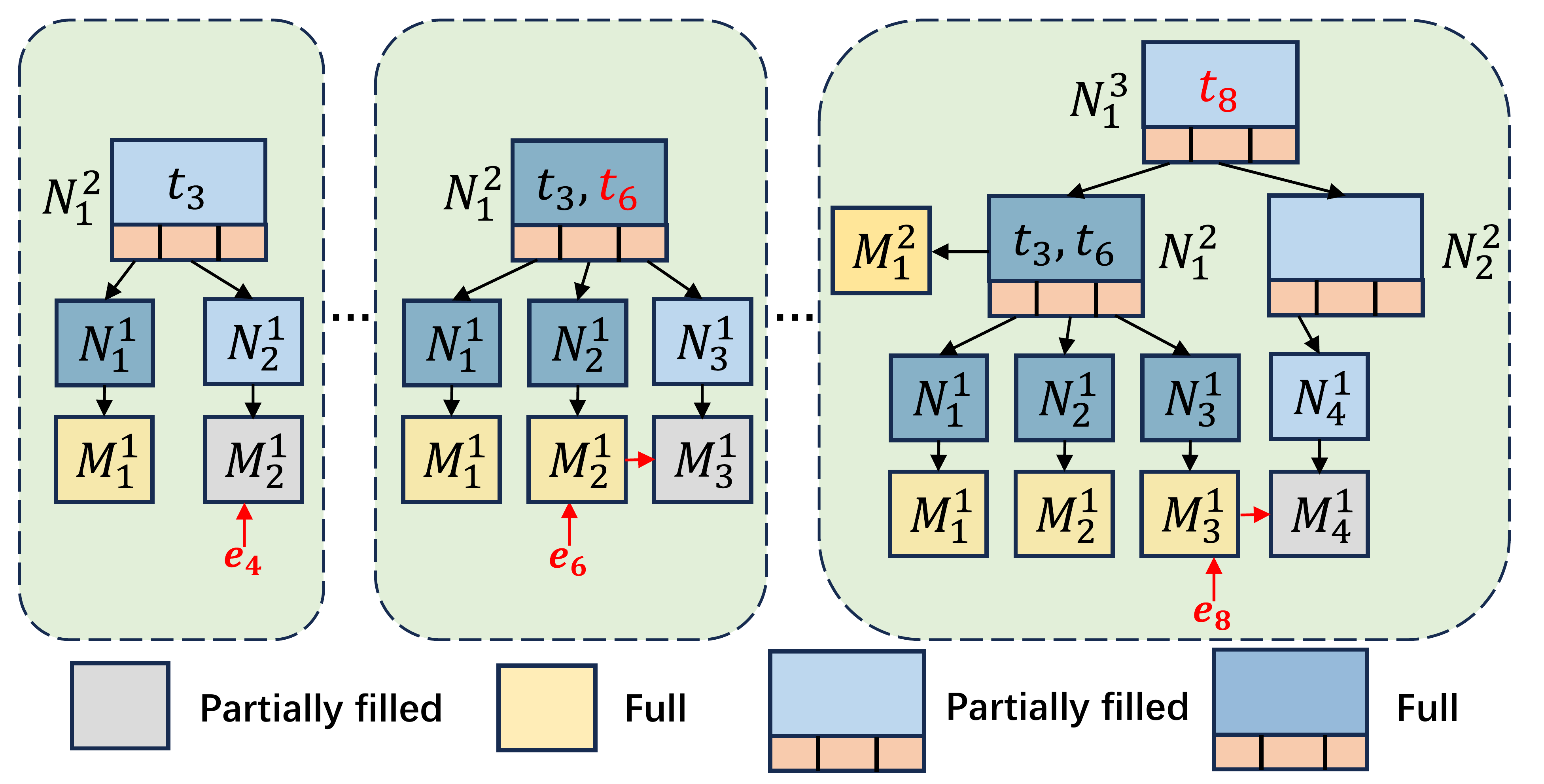}
    \vspace{-5pt}
    \caption{Evolution of the Structure of HIGGS}
    \vspace{-18pt}
    \label{fig:construct}
\end{figure}

%\textit{Example 4}:
\begin{example}
Fig.~\ref{fig:construct} illustrates the evolving HIGGS structure with the insertion of edge $e_{i} = (s_{i}, d_{i}, w_{i}, t_{i})$. Each node holds up to two keys and three children. Edge $e_{4}$ is smoothly inserted into matrix $M_{2}^{1}$, whereas edge $e_{6}$ is relocated to node $N_{1}^{2}$ due to conflicts. Similarly, edge $e_{8}$ also moves to $N_{1}^{2}$, which then reaches its capacity, therefore forming a new node $N_{1}^{3}$ to house $t_{8}$. This causes data aggregation into matrix $M_{1}^{2}$ and the creation of empty node $N_{2}^{2}$ to link $N_{1}^{3}$ and $N_{4}^{1}$.
\end{example}

\textbf{Insertion.}
Algorithm~\ref{alg:Insert} outlines the edge insertion procedure for HIGGS. Initially, HIGGS consists of a single root node, linked to an empty compression matrix $M$. When an incoming edge $e_{i} = (s_{i}, d_{i}, w_{i}, t_{i})$ arrives, it is processed using a hash function $H(\cdot)$ to obtain the hash values $H(s_{i})$ and $H(d_{i})$ for $s_{i}$ and $d_{i}$, respectively. The $F_{1}$-bit suffixes of $H(s_{i})$ and $H(d_{i})$ are then concatenated to create a fingerprint pair $(f(s_{i}), f(d_{i}))$. The remaining bits are used to compute the address pair $(h(s_{i}), h(d_{i}))$. For efficiency, we derive bit operations for calculating a vertex's fingerprint and address:
%Algorithm~\ref{alg:Insert} describes the procedure for inserting edges within HIGGS. Initially, HIGGS consists solely of a root node $N_{1}^{1}$, which is linked to an empty compression matrix $M_{1}^{1}$. Upon the arrival of an element $e_{i} = (s_{i}, d_{i}, w_{i}, t_{i})$, it is processed through a hash function $H(\cdot)$ to derive the hash values $H(s_{i})$ and $H(d_{i})$ corresponding to $s_{i}$ and $d_{i}$, respectively. Subsequently, the $F_{1}$-bit suffixes of $H(s_{i})$ and $H(d_{i})$ are concatenated to form a fingerprint pair $<f(s_{i}), f(d_{i})>$. The remaining bits are utilized to calculate the address pair $<h(s_{i}), h(d_{i})>$ by taking the module of $d_{1}$. Formally, for a node $v$, the methodology for calculating its fingerprint and address is as follows:
\begin{equation}
\small
%\text{Bit~Operation:}
        f(v) = H(v)  \&  (2^{F_{1}} - 1) \quad \text{and} \quad
        h(v) = (H(v) >> F_{1}) \% d_{1}
\label{eq:fanda}	
\end{equation}
%\begin{equation}
%    \begin{split}
%        f(v) & = H(v)  \&  (2^{F_{1}} - 1) \\
%        h(v) & = (H(v) >> F_{1}) \% d_{1}
%    \end{split}
%\label{eq:fanda}
%\end{equation}

Based on the calculated address, we can quickly locate the bucket $M[h(s_{i})][h(d_{i})]$. We then check each entry in this bucket that has a value. If an entry's fingerprint pair matches $(f(s_{i}), f(d_{i}))$ and the recorded timestamp is $t_{i}$, we increase its stored weight by $w_{i}$. If no such entry is found, we insert $(f(s_{i}), f(d_{i}), t_{i}, w_{i})$ into an empty entry. If all entries are occupied without a match, the insertion fails, prompting the creation of a new leaf node that connects to a matrix of the same size as $M$, into which the edge is inserted in the same manner. The timestamp $t_{i}$ is then transmitted to the parent node to serve as an key during queries.
%{\color{blue}If the node associated with the current timestamp passed is a non-leaf node},

%%%%%%%%%%%%%%%%%%%%%%%% zhao del 1029 %%%%%%%%%%%%%%%%%%%%%%%%%%%%%%%%
% Upon receiving the transmitted timestamp, it necessitates the execution of an aggregation operation as demonstrated in Algorithm~\ref{alg:aggregate}. Initially, a matrix $M_{k}^{l}$ is constructed, which is $\sqrt{\theta}$ times the size of the matrix at the $(l - 1)$-{th} layer (line 2). Subsequently, for all the edges stored in the matrices of child nodes, they must be processed through shift operations before being inserted into $M_{k}^{l}$ (lines 5-6). Specifically, for an edge $e$, we shift the leftmost $R$ bits of the fingerprints of both the source and destination nodes to the rightmost position of their corresponding addresses, thereby forming a new fingerprint pair and address pair for insertion. It is noteworthy that the aggregation operation does not involve the storage of timestamps.

\begin{example}
As shown in Fig.~\ref{fig:newHIGGS}, consider the edge $e_{9} = (v_{2}, v_{3}, 1, t_{9})$ from Fig.~\ref{fig:graph stream}. When insertion into $M^{1}_{4}$ fails, the timestamp $t_{9}$ propagates upward, prompting the construction of a new matrix $M^{1}_{5}$ to insert $e_{9}$. By hashing source node $v_{2}$ and destination node $v_{3}$, we obtain the fingerprint pair $(100, 110)$ and address pair $(1, 0)$. Using the address pair, we locate bucket $M_{5}^{1}[1][0]$ and insert $e_{9}$'s fingerprint pair, weight, and an offset of 0 relative to the matrix's start time.
\end{example}
% {\color{red} No need for $M^{1}_{2}$-like symbols. Make leaf nodes be consistent with current definition.}
\begin{algorithm}
\footnotesize
    \caption{\textbf{Insert ($s, d, w, t$)}}
    \label{alg:Insert}
    \begin{algorithmic}[1] % 1 代表每一行都显示行号
        \State $(f(s), f(d)) \gets$ the fingerprint pair of $s$ and $d$;
        \State $(h(s), h(d)) \gets$ the address pair of $s$ and $d$;
        \State insert $(f(s), f(d)), w$ and $t$ into $M^{1}_{n_{1}}[h(s)][h(d)]$; \Comment{$n_{1}$ denotes the number of leaf nodes}
        \If {it fails}
            \State $n_{1} \gets n_{1} + 1, l \gets 1$;
            \State construct a new matrix $M^{1}_{n_{1}}$ and insert this edge;
            \Do
                \State $u = \lceil \frac{n_{1} - 1}{\theta^{l}} \rceil$, transmit $t$ to $N_{u}^{l + 1}$;
                \If {$l \geq 2$}
                    \State Aggregate($l, n_{1} - 1$);
                \EndIf
                \State $l \gets l + 1$;
            \doWhile{transmission fails}
        \EndIf
    \end{algorithmic}
\end{algorithm}

{\bf Aggregation.} The aggregation between a parent node and its children plays a vital role in enhancing query accuracy and efficiency. The aggregation process follows a bottom-up way:
\begin{itemize}
    \item \textbf{Leaf Nodes:} Each leaf node's compressed matrix is directly computed based on raw stream items.
    \item \textbf{Non-Leaf Nodes:} Matrices of upper-level nodes aggregate those of their children, encompassing the descendant nodes within their subtrees.
\end{itemize}
Algorithm~\ref{alg:aggregate} describes the details of the aggregation process. For a node at the $(i + 1)$-th layer, we construct a $\sqrt{\theta}d_{i} \times \sqrt{\theta}d_{i}$ matrix to aggregate the $\theta$ matrices of size $d_{i} \times d_{i}$ from its children at the $i$-th layer. We observe that a larger matrix increases the number of address bits. An intuitive approach is to shift additional fingerprint bits into the address, thereby reducing the storage required for fingerprints. Binary representation, which is conducive to data compression and expansion, facilitates data aggregation through simple shift operations. To avoid introducing additional errors during the aggregation process, we maintain the matrix size as a power of two. Assuming that the aggregation process reduces the fingerprint storage by $R$ bits, increasing the address bits by $R$, the matrix size becomes $4^R$ times its original size. Therefore, $\theta$ needs to be a power of four. At the $i$-th layer, the fingerprint length is restricted to $F_{1} - (i - 1)R$ bits. Here, $\theta$ denotes the maximum number of child nodes connected to a node, and $d_{i}$ and $F_{i}$ represent the matrix size and fingerprint length at the $i$-th level, respectively. It is noteworthy that the aggregation operation does not involve the storage of timestamps.

\begin{figure}[htbp]
    \vspace{-5pt}
    \centering
    \includegraphics[width=0.9\linewidth]{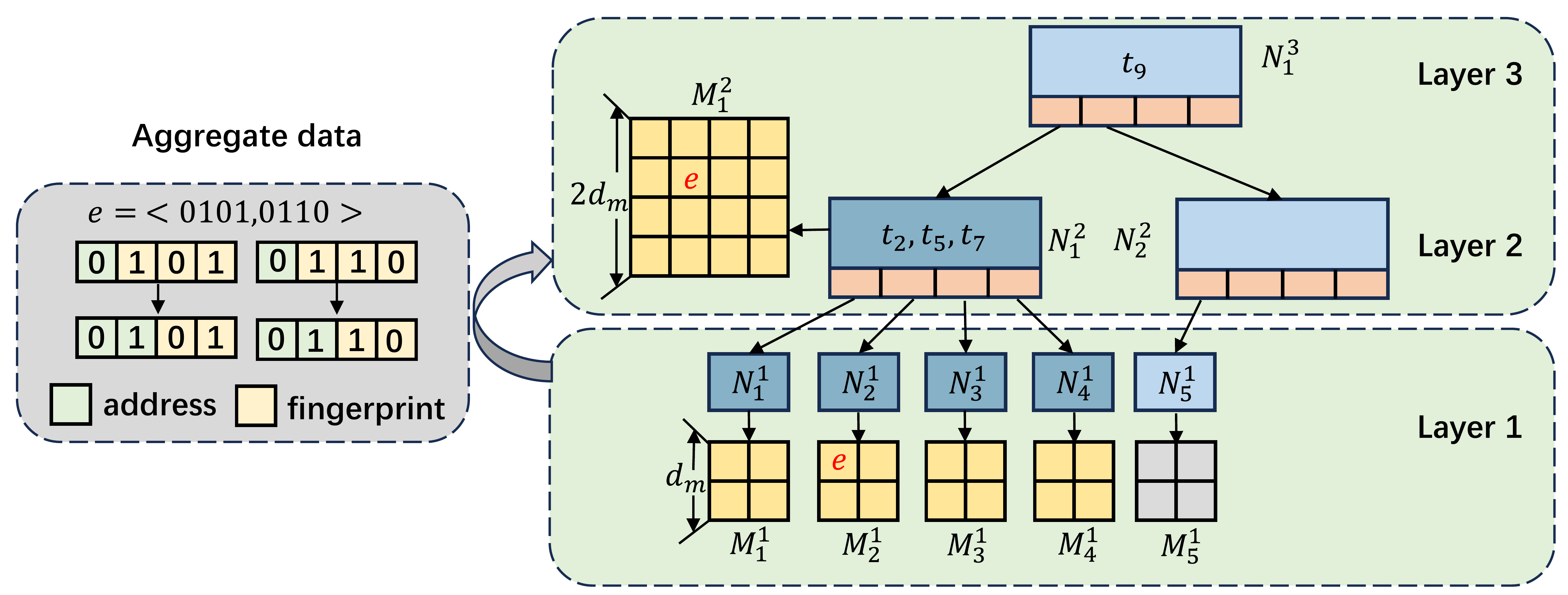}
    \caption{An Example of Aggregation}
    \label{fig:aggregate}
    \vspace{-10pt}
\end{figure}

\begin{example}
Fig.~\ref{fig:aggregate} shows an example of data aggregation with $\theta = 4$, $d_{1} = 2$, and $F_{1} = 3$. For edge $e = (0101,0110)$ in $M_{2}^{1}$, its address pair is $(0,0)$ and fingerprint pair $(101,110)$. During aggregation to $M_{1}^{2}$, with $R = 1$, the leftmost bit of the fingerprint moves to the address's rightmost bit. The new address pair $(01,01)$ locates $e$ in $M_{1}^{2}$, where fingerprint pair $(01,10)$ is stored.
\end{example}
%%%%%%%%%%%%%%%%%% zhao del 1030 %%%%%%%%%%%%%%%%%%%%%%%%%%%%%%%%%%%%%%
% Fig.~\ref{fig:aggregate} illustrates an example of aggregating data. In this instance, $\theta$ is set to 4, $d_{1}$ to 2, and $F_{1}$ to 3. For the edge $e = \langle 0101,0110 \rangle$ located at $M_{2}^{1}$, its address pair is $\langle 0,0 \rangle$, and the fingerprint pair is $\langle 101,110 \rangle$. During aggregation to $M_{1}^{2}$, given that $R = 1$, it implies that the fingerprint can be stored with one less bit. Consequently, the leftmost bit of the fingerprint is relocated to the rightmost position of the address. Following this adjustment, the new address pair $\langle 01,01 \rangle$ is utilized to pinpoint the position of $e$ in $M_{1}^{2}$, where the fingerprint pair $\langle 01,10 \rangle$ is subsequently stored.

\vspace{-5pt}
\begin{algorithm}
\footnotesize
    \caption{\textbf{Aggregate ($l, u$)}}
    \label{alg:aggregate}
    \begin{algorithmic}[1] % 1 代表每一行都显示行号
        \State $d_{l} = \sqrt{\theta}d_{l - 1}$;
        \State $k = \frac{u}{\theta ^{l - 1}}$, construct a $d_{l} \times d_{l}$ matrix $M_{k}^{l}$;
        \For {$i \gets 1$ to $\theta$}
            \For{\textbf{each} $e \in M_{(k - 1)\theta + i}^{l - 1}$} \Comment{$e$ is an edge}
                \State Obtain the fingerprint pair $(f'(s), f'(d))$ and address pair $(h'(s), h'(d))$ of $e$ at level $l$ through shift operations;
                \State insert $(f'(s), f'(d)), w$ into $M^{l}_{k}[h'(s)][h'(d)]$;
            \EndFor
        \EndFor
    \end{algorithmic}
\end{algorithm}
% \vspace{-5pt}
\textbf{TRQ Evaluation Framework.}
For any specified temporal range query, including edge, vertex, path, and subgraph queries, it can be broken down into a series of sub-range queries using the boundary search algorithm (Algorithm~\ref{alg:query}) run on HIGGS. Each sub-range query is performed on its respective compressed matrix, thereby transforming the problem into querying across a set of matrices.
Different types of query primitives, such as edge and vertex queries, correspond to different methods of accessing the compressed matrices.
The boundary search algorithm involves two main steps:

First, starting from the root node of HIGGS, the algorithm identifies the child nodes entirely covered by the queried range $[t_{s}, t_{e}]$ and adds them to the query list $X$. If $[t_{s}, t_{e}]$ falls within the range represented by a child node, further queries are made into the child node, continuing until $t_{s}$ and $t_{e}$ no longer reside within the range of the same child node.

Second, the algorithm then identifies the child nodes containing $t_{s}$ or $t_{e}$ and adds the child nodes of parent nodes that are contained by $[t_{s}, +\infty]$ or $[-\infty, t_{e}]$ to $X$. If $t_{s}$ or $t_{e}$ exactly matches the start or end point of a node's range, the corresponding query can be directly terminated; otherwise, the search extends to the next level of child nodes containing $t_{s}$ or $t_{e}$. Upon reaching the leaf nodes, in addition to adding nodes that meet the previously mentioned conditions, nodes containing $t_{s}$ or $t_{e}$ are also added to $X$. This concludes the decomposition process of the temporal range.

\begin{algorithm}
\footnotesize
    \caption{\textbf{BoundarySearch ($t_{s}, t_{e}$)}}
    \label{alg:query}
    \begin{algorithmic}[1] % 1 代表每一行都显示行号
        \State $X \gets \emptyset, j \gets 1$;
        \State $l \gets$ the number of layers in HIGGS;
        \Do
            \State $I_{x}^{l} \gets \mathit{getLeftTimeBoundary(N_{j}^{l}, t_{s})}$; \Comment{$t_{s}$ exactly covers $I_{x}^{l}$}
            \State $I_{y}^{l} \gets \mathit{getRightTimeBoundary(N_{j}^{l}, t_{e})}$;
            \If {$I_{x}^{l} \leq I_{y}^{l}$}
                \State $X \gets X \cup \mathit{getEntities(N_{j}^{l}, I_{x}^{l}, I_{y}^{l})}$; \Comment{get the subranges contained by $[I_{x}^{l}, I_{y}^{l})$ in $N_{j}^{l}$ and add them to $X$}
                \State $lt \gets j - 1 + x, rt \gets j + y$;
                \State break;
            \EndIf
            \State $l \gets l - 1, j \gets j - 1 + x$;
        \doWhile{true}
        \State $p \gets q \gets l - 1, t_{u} \gets I_{x}^{l}, t_{v} \gets I_{y}^{l}$;
        \While{$p \geq 1$ and $t_{s} \neq t_{u}$}
            \State $I_{x}^{p} \gets \mathit{getLeftTimeBoundary(N_{lt}^{p}, t_{s})}, t_{u} \gets I_{x}^{p}$;
            \State $X \gets X \cup \mathit{getEntities(N_{lt}^{p}, I_{x}^{p}, +\infty)}$;
            \State $lt \gets lt - 1 + x$;
            \State $p \gets p - 1$;
        \EndWhile
        \While{$q \geq 1$ and $t_{e} \neq t_{v}$}
            \State $I_{y}^{q} \gets \mathit{getRightTimeBoundary(N_{rt}^{q}, t_{e})}, t_{v} \gets I_{y}^{q}$;
            \State $X \gets X \cup \mathit{getEntities(N_{rt}^{q}, -\infty, I_{y}^{q})}$;
            \State $rt \gets rt + y$;
            \State $q \gets q - 1$;
        \EndWhile
        \State return $X$
    \end{algorithmic}
\end{algorithm}
\vspace{-5pt}

\begin{example}
Fig.~\ref{fig:decompose} illustrates how a boundary search algorithm decomposes a range within a four-level HIGGS structure ($\theta = 4$). For time range $[t_{s}, t_{e}]$, starting from $N_{1}^{4}$, the search proceeds to nodes $N_{1}^{3}$ (containing $t_{s}$) and $N_{2}^{3}$ (containing $t_{e}$). Since no child nodes are fully within $[t_{s}, t_{e}]$, children of $N_{1}^{3}$ ($N_{3}^{2}$, $N_{4}^{2}$) and child $N_{8}^{1}$ of $N_{2}^{2}$ are queried and added to list $X$. Leaf node $N_{7}^{1}$, containing $t_{s}$, is also added. The process is mirrored for $t_{e}$, resulting in query list $X$ containing matrices $M_{7}^{1}, M_{8}^{1}, M_{3}^{2}, M_{4}^{2}, M_{17}^{1}, M_{18}^{1}$, and $M_{19}^{1}$.
%%%%%%%%%%%%%%%%%%%%%% zhao del 1030 %%%%%%%%%%%%%%%%%%%%%%%%%%%%%%%%%%%%%%%%%
% Fig.~\ref{fig:decompose} shows how a boundary search algorithm decomposes a range within a HIGGS structure containing four levels ($\theta = 4$). For the time range $[t_{s}, t_{e}]$, where $I_{5}^{2} \leq t_{s} < I_{6}^{2}$ and $I_{14}^{2} \leq t_{e} < I_{15}^{2}$, the search starts at $N_{1}^{4}$ and moves to $N_{1}^{3}$, which includes $t_{s}$, and $N_{2}^{3}$, which includes $t_{e}$. As none of the child nodes at $N_{1}^{4}$ are fully within $[t_{s}, t_{e}]$, the query further involves $N_{1}^{3}$'s children, $N_{3}^{2}$ and $N_{4}^{2}$, and $N_{2}^{2}$’s child, $N_{8}^{1}$, all meeting the criteria and added to query list $X$. Leaf node $N_{7}^{1}$, containing $t_{s}$, is also added. The process for $t_{e}$ mirrors $t_{s}$’s, resulting in query list $X$ including matrices $M_{7}^{1}, M_{8}^{1}, M_{3}^{2}, M_{4}^{2}, M_{17}^{1}, M_{18}^{1}$, and $M_{19}^{1}$.
\end{example}

% Given that matrix $M_{1}^{3}$ is entirely encompassed by $[t_{s}, t_{e}]$, it can be directly added to the query list. Subsequent queries move to the child node $N_{2}^{2}$ containing $t_{e}$. As no matrix is fully included within $[I_{1}^{3},t_{e}]$, the search extends to the subsequent child node $N_{5}^{1}$ that contains $t_{e}$. At this juncture, both matrices $M_{17}^{1}$ and $M_{18}^{1}$, being completely enclosed within $[I_{1}^{3},t_{e}]$, are added to the query list. Since $N_{5}^{1}$ is a leaf node, matrix $M_{19}^{1}$, which contains $t_{e}$, also necessitates inclusion. This concludes the range decomposition process, with the query list comprising matrices $M_{1}^{3}$, $M_{17}^{1}$, $M_{18}^{1}$, and $M_{19}^{1}$.

\begin{figure}[htbp]
    \centering
    \vspace{-16pt}
    \includegraphics[width=0.9\linewidth]{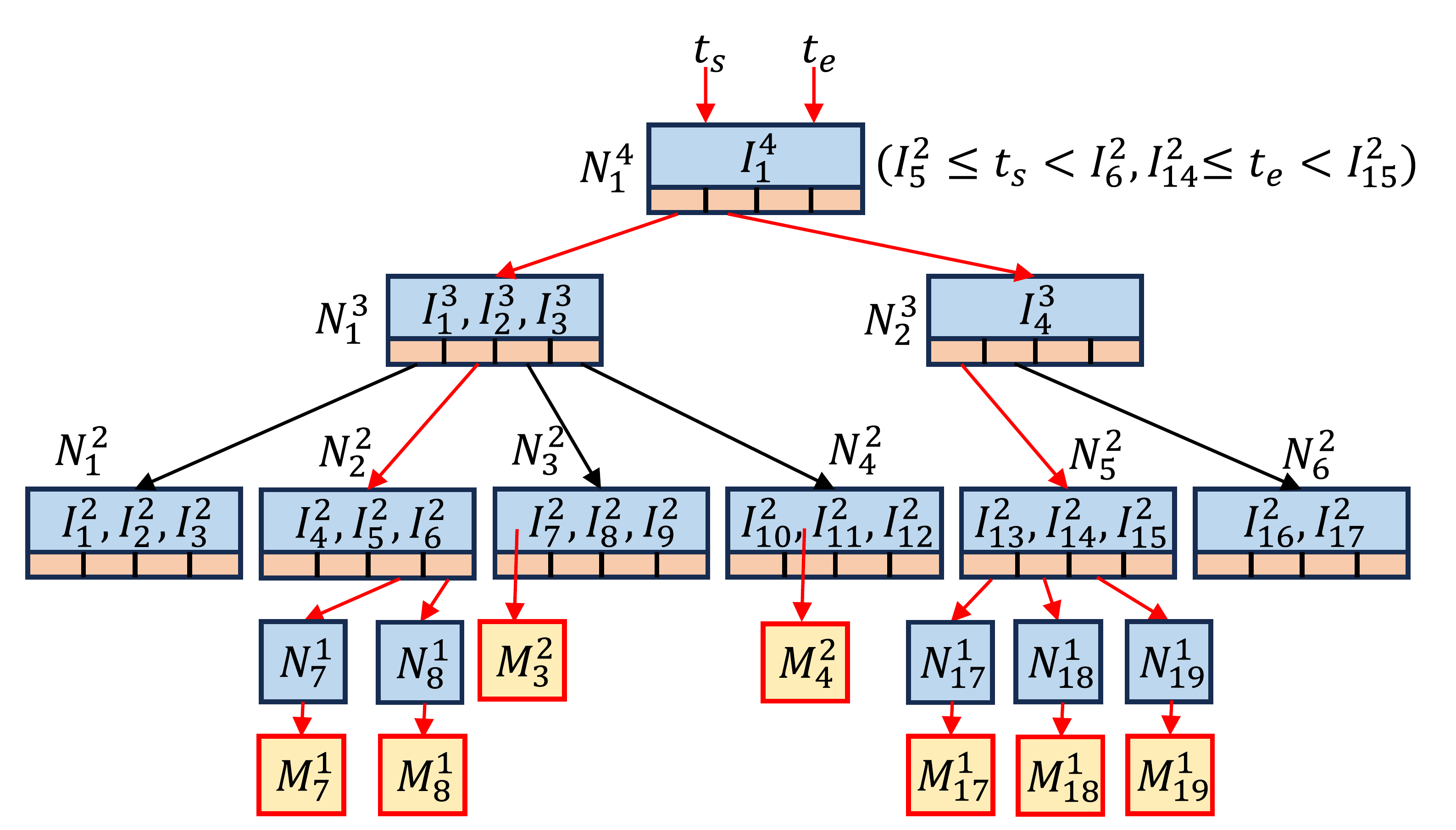}
    \vspace{-2pt}
    \caption{An Example of Range Decomposition using the Boundary Search Algorithm}
    \vspace{-10pt}
    \label{fig:decompose}
\end{figure}

%Next, we need to query each matrix in the query list. Given that time range query primitives primarily consist of node and edge queries, matrix-based queries correspondingly encompass both types. We will now elaborate on their details.

Next, we proceed to discuss how the query framework is implemented for TRQ primitives, i.e., edge and vertex queries.

\textbf{Edge Queries.} For a temporal query range $[t_{s}, t_{e}]$, an edge query from $s$ to $d$ is conducted on matrix $M$. Using Formula~\eqref{eq:fanda} and a shift operation, we derive the fingerprint $(f(s),f(d))$ and address $(h(s),h(d))$, locating the bucket $M[h(s)][h(d)]$.

Two cases arise: 1) Non-leaf node: If $M$ is a non-leaf, we check if any entry's fingerprint matches $(f(s),f(d))$—if so, return its weight; otherwise, it returns $0$. 2) Leaf node: If $M$ is a leaf, we additionally check if the timestamp falls within $[t_{s}, t_{e}]$. If both the fingerprint and timestamp match, the entry is considered valid, and we proceed as in the non-leaf case.

%%%%%%%%%%%%%%%%%%%%%%% zhao del 1030 %%%%%%%%%%%%%%%%%%%%%%%%%%%%%%%%%%%%%%%
% For a specific temporal query range $[t_{s}, t_{e}]$, an edge query from $s$ to $d$ is conducted on matrix $M$. Utilizing Formula~\eqref{eq:fanda} and a shift operation, we can derive the fingerprint pair $\langle f(s),f(d)\rangle$ and the address pair $\langle h(s),h(d)\rangle$, corresponding to $s$ and $d$ within the matrix, and then locate the bucket $M[h(s)][h(d)]$. There are two cases that may arise. 1) Non-leaf node case: if $M$ belongs to a non-leaf node, we only need to verify whether any entry's fingerprint pair matches $\langle f(s),f(d)\rangle$. If a matching entry is found, its weight is returned as the result; otherwise, a result of $0$ is returned to indicate no match. 2) Leaf node case: if $M$ belongs to a leaf node, an additional check on whether the timestamp is within the temporal query range $[t_{s}, t_{e}]$ is required.
% Only if both the fingerprint pair and the timestamp match is it considered a valid entry, and the subsequent process is the same as in the non-leaf node scenario.

\textbf{Vertex Queries.} For the temporal query range $[t_{s}, t_{e}]$, queries are made on a source or destination vertex $v$ based on matrix $M$ of size $d'$. Using Formula~\eqref{eq:fanda} and a shift operation, we ascertain $v$'s fingerprint $f(v)$ and address $h(v)$, then locate the corresponding row in $M$:
\begin{equation}
\small
    \hspace{-5pt} M[h(v)][:] = {M[h(v)][0], M[h(v)][1], ..., M[h(v)][d' - 1]}
\end{equation}
Each entry in the row's buckets is traversed. For non-leaf nodes, if the source vertex's fingerprint matches $f(v)$, the weight is added to the result. For leaf nodes, both the fingerprint must match and the timestamp must fall within $[t_{s}, t_{e}]$ for the weight to be included.
%%%%%%%%%%%%%%%%% zhao del 1030 %%%%%%%%%%%%%%%%%%%%%%%%%%%%%%%%%%%%%%%%%%%%%%
% Within the specific temporal query range $[t_{s}, t_{e}]$, queries are made on a source (or destination) vertex $v$ based on the given matrix $M$ of size $d'$.
% Again, using Formula~\eqref{eq:fanda} and a shift operation, we ascertain $v$'s fingerprint $f(v)$ and address $h(v)$ within $M$, subsequently locating the corresponding row:
% \begin{equation}
% \small
% 	\hspace{-5pt} M[h(v)][:] = {M[h(v)][0], M[h(v)][1], ..., M[h(v)][d' - 1]}
% \end{equation}
% Each entry within the row's buckets is then traversed. Similar to edge queries, this process is divided into cases for leaf and non-leaf vertices. For non-leaf vertices, if the source vertex's fingerprint matches $f(v)$, the weight is accumulated into the result. For leaf vertices, in addition to the fingerprint matching, the corresponding timestamp must also fall within the range $[t_{s}, t_{e}]$. Only if both conditions are satisfied, the weight is added to the result.

\subsection{Optimization}
\label{subsec:opt}
We propose three optimizations: {\it multiple mapping buckets} to enhance space efficiency, {\it overflow blocks} to improve accuracy, and {\it parallelization} to boost insertion throughput.

\textbf{Multiple Mapping Buckets (MMB)}.
When inserting an edge $e$, using only a single mapping bucket can lead to severe conflicts, potentially causing the insertion to fail and wasting significant space if the matrix is underutilized. To address this issue, we employ multiple mapping buckets, giving each edge multiple potential insertion positions. Specifically, we use the linear congruence method \cite{l1999tables} to generate address sequences $\{h_{i}(s)|1 \leq i \leq r\}$ for the source vertex $s$ and $\{h_{j}(d)|1 \leq j \leq r\}$ for the destination vertex $d$, where $r$ represents the number of mapping positions of a vertex.
The address sequences of the source and destination vertices are then combined pairwise to produce $r \times r$ mapping buckets. To track the position of each bucket in the address sequence, we record an index pair $(i,j)$. In this way, insertion only fails if all mapping buckets conflict, significantly reducing the likelihood of failure and improving space utilization of the matrix.

\textbf{Overflow Blocks (OB).} In the basic HIGGS framework, if an edge insertion at a leaf node fails, a new leaf is created to store the edge, and its timestamp is transmitted to the parent as a key to distinguish between leaf nodes. However, if multiple edges with identical timestamps have already been stored, this key becomes ineffective, leading to errors. To resolve this, we employ overflow blocks to aggregate edges that arrive simultaneously into a unified time range. If an edge insertion fails and has the same timestamp as the previous edge, an overflow block is created on the current leaf to store it; otherwise, a new leaf is created. The overflow block is essentially a small-scale compressed matrix, enabling more precise partitioning of the graph stream by timestamps and thereby enhancing query accuracy.

\textbf{Parallelization}. To enhance throughput, we adopt a parallel updating strategy, assigning each layer a separate thread that updates only the latest node. To maintain data consistency across the subtree, each element is first updated by the thread corresponding to leaf nodes before being processed by other threads. Since order preservation is needed only at the element level, parallel efficiency remains high. By utilizing parallelism, the throughput of HIGGS is increased substantially.

\section{ANALYSIS}
\subsection{Space Cost Analysis}
%Next, we will detail the space overhead of the HIGGS by examining three key aspects: the space savings from the aggregation algorithm, the utilization rate of the matrix, and the overall complexity.

\textbf{Space Savings through Aggregation.} Compared to the storage methods employed in existing works, our design significantly reduces space costs. %Formally, this leads to the following theorem.

%\textit{Theorem 1}:
\begin{theorem}
Compared to the storage methods of the existing proposals, a HIGGS structure with $l$ layers reduces the space cost by a proportion of $\frac{(l - 1)R}{\beta}$, where $\beta$ denotes the size of each entry within the matrix.
\end{theorem}
\begin{proof}
%\textit{Proof}.
Consider a HIGGS structure with $k$ leaf nodes, each with a compressed matrix. For simplicity, we assume that $k$ is divisible by $\theta$. In this setting, using the storage method of the previous framework, the space consumed is:
\begin{equation}
\small
    M_{1} = bkl \beta d_{1}^{2}
\end{equation}
Through the aggregation of matrices and the compression of fingerprints, the saved space is:
\begin{equation}
\small
    M_{2} = \sum_{i = 0}^{l - 1} 2bkiRd_{1}^{2} = 2bkRd_{1}^{2}\sum_{i = 0}^{l-1} i = bkRd_{1}^{2}l(l-1)
\end{equation}
Thus, the ratio of space saved is:
\begin{equation}
\small
    \frac{M_{2}}{M_{1}} = \frac{bkRd_{1}^{2}l(l-1)}{bkl \beta d_{1}^{2}} = \frac{R(l-1)}{\beta}
\end{equation}
Thus, the theorem is proven.
\end{proof}

\vspace{-5pt}
\textbf{Matrix Utilization Rate.}
%%%%%%%%%%%%%%%%% zhao del 1030 %%%%%%%%%%%%%%%%%%%%%%%%%%%%%
% The memory cost of HIGGS includes the space consumed by keys and matrices, with matrices occupying the substantial majority. Therefore, the efficient utilization of matrix space is crucial for the scalability of HIGGS.
The utilization rate of a matrix refers to the proportion of elements stored within the matrix relative to its total capacity. As discussed in Section~\ref{subsec:opt}, for a $d \times d$ compressed matrix,
we enhance the utilization rate of the bucket by setting the number of mapping buckets for each edge to $p$. Simultaneously, each bucket contains $b$ entries.
An insertion failure indicates that all $p$ mapping buckets have encountered conflicts, while all previous edges were successfully inserted. Let $A_{i}$ denote the probability that $i$-{th} edge is inserted successfully, and $1 - A_{i}$ represent the probability of its failure. If $k$-{th} edge causes the first insertion failure, the probability of this occurrence, according to the geometric distribution, is:
%An edge insertion fails, it means all $p$ mapping buckets have encountered a conflict and all previous edges were successfully inserted. Let the event $A_{i}$ denote the probability that the $i^{th}$ edge is inserted successfully, and $1 - A_{i}$ as the probability of its failure. If the $k^{th}$ edge causes the first insertion failure, according to the geometric distribution, this probability is,
\begin{align}
\small
    \begin{split}
        Pr(X = k) & = \prod_{i = 1}^{k - 1} A_{i} \times (1 - A_{k}) \\
                  & = \prod_{i = 1}^{k - 1} (1 - (\frac{i - 1}{bd^{2}})^{bp}) \times (\frac{k - 1}{bd^{2}})^{bp}
    \end{split}
\end{align}
Given that the compressed matrix containing $b \times d \times d$ elements, the expected utilization $E(\alpha)$ of the matrix is calculated as follows:

\begin{equation}
\small
    E(\alpha) = \frac{E(k)}{bd^{2}} = \frac{\sum_{k = 1}^{bd^{2}} k \cdot Pr(X = k)}{bd^{2}}
\end{equation}
Here, $E(k)$ represents the expected number of elements successfully inserted into the matrix. Therefore, the formula quantifies the average utilization of the storage block by considering the successful insertions relative to the total number of available elements of its associated compressed matrix.

\textbf{Space Complexity Analysis.}
The number of layers in HIGGS primarily depends on the quantity of leaf nodes and the out-degree of each node. Since the space occupied by a compressed matrix to which a non-leaf node points equals the cumulative space of the matrices of all its corresponding child nodes, the space overhead for each layer in HIGGS is fundamentally consistent.
Let the average space utilization rate of each matrix pointed to by the leaf nodes be $\alpha$, and let the space occupied by the keys in HIGGS be $I$. The number of matrices at the leaf nodes is given by $n_{1} = \frac{|E|}{\alpha b d_{1}^{2}}$, leading to a space complexity for HIGGS of approximately $O(|E| \log \frac{|E|}{\alpha b d_{1}^{2}} + I)$. Since the space required for storing keys is significantly less than that needed for storing graph stream data, the overall space cost is $O(|E| \log \frac{|E|}{\alpha b d_{1}^{2}})$. After optimization with overflow blocks, assuming the average time span represented by each matrix is $L'$ (where $1 \leq L' = \frac{\alpha b L d_{1}^{2}}{|E|} \leq L$), the formula can be derived as follows:
\begin{equation}
\small
    E\log n_{1} = E\log (\frac{L}{L'}) \leq E\log L
\end{equation}
$O(E \log L)$ represents the space complexity of the existing framework. This establishes that the space complexity of HIGGS is superior to that of the existing framework. Additionally, the fewer edges arrive per unit time in the graph stream, that is, the larger the value of $L'$, the lower the space cost becomes.

\subsection{Query Efficiency Analysis}
The time complexity of HIGGS depends on the number of matrices accessed at each layer. Assuming the length of a given temporal range query is $L_{q}$, at most $2(\theta - 1) \log_{\theta} \frac{L_{q}}{L'}$ matrices are queried, leading to a time complexity approximates O($\log \frac{L_{q}}{L'}$). %Additionally, when fewer edges arrive within a unit of time, meaning the average time span $L'$ represented by leaf nodes is larger, the query time will further decrease.

Furthermore, the query at each layer in the existing framework targets the data structure that stores the entire graph stream, resulting in larger scales. Consequently, when addressing temporal range queries related to topological structures, this framework accesses a greater number of buckets, which leads to lower query efficiency. In contrast, the total size of all matrices to be queried in HIGGS is consistent with the size of the graph stream within that time range, making it smaller in scale and more efficient to access.

% \vspace{-10pt}

\subsection{Scalability Analysis}
As streaming edges continuously arrive, HIGGS increases the number of leaf nodes to store data while aggregating at non-leaf nodes. The number of layers in HIGGS is $O(\log n_{1})$, with a space complexity of $O(E\log n_{1})$, an update time complexity of $O(E\log n_{1})$, and a query time complexity of $O(\log \frac{L_{q}}{L'})$, where $n_{1}$ is the number of leaf nodes, $L_{q}$ is the query range length, and $L'$ is the average time range length represented by each leaf node. To accelerate updates, we employ parallel optimization (see Section~\ref{subsec:opt}), where each layer is assigned a dedicated thread. %As the number of layers increases, an additional thread is allocated, reducing the update complexity to $O(E)$.} 

\subsection{Query Accuracy Analysis}

\textbf{Collision Rate Analysis.} In this section, we analyze the collision rates for nodes and edges maintained in HIGGS.
Moreover, we show that HIGGS has one-sided error, i.e., it only overestimates the true query results; it never underestimates results.

{\bf 1)} \underline{\it{Node Collision.}} For a given temporal query range \([t_{s},t_{e}]\) and an edge \(e_{i}=(s_{i}, d_{i}, w_{i}, t_{i})\), conflicts for \(e_{i}\)'s source (or destination) node arise only from other edges within this time range.
Specifically, if there exists an edge \(e_{j} = (s_{j}, d_{j}, w_{j}, t_{j})\) conflicting with \(e_{i}\), the following conditions must be met:
\begin{itemize}
\small
    \item \(s_{i} \neq s_{j}\) (or \(d_{i} \neq d_{j}\));
    \item \(h(s_{i}) = h(s_{j})\) (or \(h(d_{i}) = h(d_{j})\));
    \item \(t_{s} \leq t_{i}, t_{j} \leq t_{e}\).
\end{itemize}
%(1) \(s_{i}\neq s_{j}\) (or \(d_{i} \neq d_{j}\) );
%(2) \(h(s_{i}) = h(s_{j})\) (or \(h(d_{i}) = h(d_{j})\));
%(3) \(t_{s} \leq t_{i}, t_{j} \leq t_{e}\).
Suppose that within the time range $[t_{s}, t_{e}]$, the number of source (or destination) nodes different from $s_{i}$ (or $d_{i}$) is denoted as $k$. The total number of distinct source (or destination) nodes in the entire graph stream is $K$, and the size of the value range of the hash function is $Z$, where $Z= d_{1} 2^{F_{1}}$. The probability of a node conflicting with $s_{i}$ is $\frac{1}{Z}$. Given that there are $k$ such nodes within $[t_{s}, t_{e}]$, the node conflict probability for $e_{i}$ is,
%Suppose that within the time range \([t_{s},t_{e}]\), the number of source (resp. destination) nodes different from \(s_{i}\) (resp. \(d_{i}\)) is denoted as \(k\). The total number of distinct source (resp. destination) nodes in the entire graph stream is \(K\), and the codomain size of the hash function is \(Z\), where \(Z= d_{1} 2^{F_{1}}\), then the probability of a node conflicting with $s_{i}$ is \(\frac{1}{Z}\). Given that there are $k$ such nodes within \([t_{s},t_{e}]\), the node conflict probability for \(e_{i}\) is,
\begin{align}
\label{eq:node}
\small
    \begin{split}
        Pr(node) & = 1 - (1 - \frac{1}{Z})^{k} \approx 1 - e^{-\frac{k}{Z}}  \\
                 & \leq 1 - e^{-\frac{K}{Z}} = 1 - e^{-\frac{K}{d_{1} 2^{F_{1}}}}
    \end{split}
\end{align}
From the equation, it is evident that increasing the number of fingerprint bits $F_{1}$ or the size of the compression matrix $d_{1}$ in the first layer will reduce the probability of node conflicts.
%From the Formula~\eqref{eq:node}, it is evident that if we set the number of fingerprint bits \(F_{1}\) to a larger value or increase the size of the compression matrix \(d_{1}\) in the main layer, both scenarios will lead to a reduction in the probability of node conflicts.

%\textbf{Edge collision.} Similar to Node collision, when provided with a specified time range \([t_{s},t_{e}]\) and an edge \(e_{i}=(s_{i}, d_{i}, w_{i}, t_{i})\) for query, edges outside this time range will not conflict with \(e_{i}\). If there is an edge $e_{j}$ that conflicts with $e_{i}$, then it must satisfy the following condition:
{\bf 2)} \underline{\it{Edge Collision.}} Similar to node collision, when provided with a specified time range $[t_{s}, t_{e}]$ and an edge $e_{i}=(s_{i}, d_{i}, w_{i}, t_{i})$ for a query, edges outside beyond the time range will not conflict with $e_{i}$. If there is an edge $e_{j}$ that conflicts with $e_{i}$, it must satisfy the following conditions:
%(1) $s_{i} \neq s_{j}$ or $d_{i} \neq d_{j}$;
%(2) $h(s_{i}) = h(s_{j})$ and $h(d_{i}) = h(d_{j})$;
%(3) $t_{s} \leq t_{i}, t_{j} \leq t_{e}$.
\begin{itemize}
\small
    \item $s_{i} \neq s_{j}$ or $d_{i} \neq d_{j}$
    \item $h(s_{i}) = h(s_{j})$ and $h(d_{i}) = h(d_{j})$
    \item $t_{s} \leq t_{i}, t_{j} \leq t_{e}$
\end{itemize}
%%%%%%%%%%%% zhao del 1030 %%%%%%%%%%%%%%%%%%%%%%%
% Then, the probability that \(e_{i}\) encounters collisions is as follows. 
Then, the collision probability of edge $e_{i}$ is:
\begin{align}
\small
    \begin{split}
        Pr(edge) & = 1 - (1 - \frac{1}{Z})^{S} \cdot (1 - \frac{1}{Z^2})^{C' - S}                          \\
                 & \approx 1 - e^{-\frac{(Z-1)S + C'}{Z^2}}                                                   \\
                 & \leq 1 - e^{-\frac{(Z-1)\max(\Phi_{o}, \Phi_{i}) + C}{Z^2}}                                      \\
                 & = 1 - e^{-\frac{(d_{1}2^{F_{1}}-1)\max(\Phi_{o}, \Phi_{i}) + C}{d_{1}^{2}4^{F_{1}}}}
    \end{split}
\end{align}Here, $S$ denotes the number of edges that, within the range $[t_{s}, t_{e}]$, share the same source or destination node as $e_{i}$ and have a probability of conflict with $e_{i}$ of $\frac{1}{Z}$. Conversely, the probability of conflict for the remaining edges, which are distinct from $e_{i}$, is $\frac{1}{Z^2}$. Additionally, $C'$ denotes the number of edges distinct from $e_{i}$ within this range. $\Phi_{o}$ and $\Phi_{i}$ respectively signify the maximum out-degree and in-degree across the entire graph stream, while $C$ denotes the count of distinct edges throughout the graph stream. Similar to node collision, increasing the fingerprint size or matrix size in the first layer proves effective in reducing the probability of conflicts.

\textbf{Error Bounds.}
%In this section, we analyze the error bound for node and edge queries in HIGGS, as well as the configuration of parameters \(F_{1}\) and \(d_{1}\).
Next, we analyze the error bounds for node and edge queries in HIGGS and the configuration of the fingerprint length $F_{1}$ and the matrix size $d_{1}$ of the leaf nodes.

{\bf 1)} \underline{\it{Error Bound for Vertex Queries.}} We use $\hat{w}_{node}$ and $w_{node}$ to represent the estimated value and the actual value for vertex queries within the range $[t_{s},t_{e}]$ of length $L_{q}$, respectively. Given a parameter $\varepsilon$, we set $F_{1} = \log \left(\frac{e}{d_{1}\varepsilon}\right)$, resulting in $Z = \frac{e}{\varepsilon}$. We denote the number of edges within the range $[t_{s},t_{e}]$ as $E'$ and the sum of their weights as $\|\mathit{w}\|'$.

%%%%%%%%%%%%%%%%%% zhao del 1023 %%%%%%%%%%%%%%%%%%%%%%%%%%%%%%%%
% , with the total sum of edge weights for the entire graph stream represented as $\|\mathit{w}\|$.

%\textbf{THEOREM 5.1:}
\vspace{-5pt}

\begin{theorem}
\label{thm:erbnode}
The result $\hat{w}_{node}$ has the guarantee: $\hat{w}_{node} \leq {w}_{node} + \varepsilon \| \mathit{w} \|'$ with a probability of at least \(1 - e^{-1}\).
\end{theorem}
\vspace{-5pt}
\begin{proof}
Consider two edges, \((x, y, w_{0}, t_{0})\) and \((s_{1}, d_{1}, w_{1}, t_{1})\), within the same time range \([t_{s},t_{e}]\). We introduce an indicator \(I_{x, s_{1}}\) which is one if there is a source node collision between the two edges, and zero otherwise. We can have the expectation of \(I_{x, s_{1}}\) as follows.
\begin{equation}
\small
    E(I_{x, s_{1}}) = Pr[h(x) = h(s_{1})] = 1/Z = \varepsilon/e
\end{equation}
%\hl{where \(e\) is Euler's number.}
Let \(X_{x}\) denote \(\sum_{i = 1}^{E'} I_{x, s_{i}}w_{i}\). Then, it follows that:\(\hat{w}_{node} = {w}_{node} + X_{x}\). By linear expectation,
\vspace{-5pt}
\begin{align}
\small
    \begin{split}
        E(X_{x}) & = E(\sum_{i = 1}^{E'} I_{x, s_{i}}w_{i}) \leq \sum_{i = 1}^{E'} w_{i}E(I_{x, s_{i}}) \\
                 & \leq (\varepsilon/e)\| \mathit{w} \|'
    \end{split}
    \vspace{-5pt}
\end{align}
By applying the Markov inequality, we have:
\begin{align}
\small
    \begin{split}
          & Pr[\hat{w}_{node} > {w}_{node} + \varepsilon \| \mathit{w} \|']     \\
        = & Pr[{w}_{node} + X_{x} > {w}_{node} + \varepsilon \| \mathit{w} \|'] \\
        = & Pr[X_{x} > \varepsilon \| \mathit{w} \|'] \leq Pr[X_{x} > eE(X_{x})] < e^{-1}
    \end{split}
\end{align}
Thus, Theorem~\ref{thm:erbnode} is proven.
The destination vertex query can be analyzed in a similar way.
%The destination node query can be analyzed by a similar discussion, which has the similar guarantee.
\end{proof}

\vspace{-5pt}
{\bf 2)} \underline{\it{Error Bound for Edge Queries.}} Similar to vertex queries, We use $\hat{w}_{edge}$ and ${w}_{edge}$ to respectively denote the estimated value and the actual value for edge queries within the range \([t_{s},t_{e}]\) of length \(L_{q}\).

%\textbf{THEOREM 5.2:}

\begin{theorem}
\label{thm:erbedge}	
The result $\hat{w}_{edge}$ guarantees $\hat{w}_{edge} \leq {w}_{edge} + \varepsilon^2 \| \mathit{w} \|'/e$ with a probability of at least \(1 - e^{-1}\).
\end{theorem}
%\textbf{Proof.}
\begin{proof}
Consider two edges, \((x, y, w_{0}, t_{0})\)and \((s_{1}, d_{1}, w_{1}, t_{1})\), within the same time range \([t_{s},t_{e}]\). We introduce an indicator \(I_{(x,y),(s_{1},d_{1})}\), which is 1 if there is an edge collision between the two edges, and is 0, otherwise. We then have the expectation of \(I_{(x,y),(s_{1},d_{1})}\):
\begin{align}
\small
    \begin{split}
        E(I_{(x,y),(s_{1},d_{1})}) & = Pr[h(x) = h(s_{1})]Pr[h(y) = h(d_{1})] \\
                                   & = 1/Z^{2} = \varepsilon^{2}/e^{2}
    \end{split}
\end{align}
Let \(X_{x,y}\) denote \(\sum_{i = 1}^{E'} I_{(x,y),(s_{i},d_{i})}w_{i}\). By construction, \(\hat{w}_{edge} = {w}_{edge} + X_{x,y}\). By linear expectation, we have:
\begin{align}
\small
    \begin{split}
        E(X_{x,y}) & = E(\sum_{i = 1}^{E'} I_{(x,y), (s_{i},d_{i})}w_{i})                                     \\
                   & \leq \sum_{i = 1}^{E'} w_{i}E(I_{(x,y), (s_{i},d_{i})})                                  \\
                   & \leq (\varepsilon^{2}/e^{2})\| \mathit{w} \|'
    \end{split}
\end{align}
By applying the Markov inequality, we have:
\begin{align}
\small
    \begin{split}
          & Pr[\hat{w}_{edge} > {w}_{edge} + \varepsilon^{2} \| \mathit{w} \|'/e]           \\
        = & Pr[{w}_{edge} + X_{x,y} > {w}_{edge} + \varepsilon^{2} \| \mathit{w} \|'/e]     \\
        = & Pr[X_{x,y} > \varepsilon^{2} \| \mathit{w} \|'/e] \leq Pr[X_{x,y} > eE(X_{x,y})] < e^{-1} \hspace{-4pt}
    \end{split}
\end{align}
Thus, Theorem~\ref{thm:erbedge} is proven.
\end{proof}

\section{EXPERIMENTAL STUDY}
\label{sec:ret}

We have implemented HIGGS and made the code available as open source on GitHub\footnote{\url{https://anonymous.4open.science/r/HIGGS-1215/}}. We report on experiments with HIGGS on real-world graph streams to answer the following research questions.
%{\color{red} some numbers, like queries generated for testing, can be reported in configuration section.}

%In our experimental design aimed at evaluating accuracy, runtime efficiency, throughput, and space cost, we formulated the following key research questions and conducted six sets of experiments to address them:

\begin{itemize}
\item[\textbf{Q1.}]
Can HIGGS outperform SOTA methods, including PGSS, Horae, Horae-cpt, AuxoTime, and AuxoTime-cpt, in terms of the query accuracy, efficiency, and scalability, for different types of graph queries, including edge, node, subgraph, and path queries?
%How does the type of query (edge, node, subgraph, path) affect the system performance? We generated a total of 210,000 queries (100,000 edge, 20,000 node, 10,000 subgraph, and 100,000 path queries) across all datasets.

\item[\textbf{Q2.}] What is the space overhead, deletion throughout, insertion throughput and latency of all competitors in achieving the query performance as aforementioned?

\item[\textbf{Q3.}] What is the effect of proposed optimization techniques, such as multiple mapping buckets, parallelization, and overflow blocks, particularly concerning space overhead, throughput and accuracy, respectively?
%{\color{red}What is the effect of proposed optimization techniques, such as xxxx and xxxx, particularly concerning xxxx and xxxx efficiency?}
%\textbf{Performance Variation of HIGGS Before and After Optimization:} How does the performance of HIGGS change after optimization?

%\item[\textbf{Q4.}] %\textbf{Sensitivity Analysis of Parameters:}
%How sensitive is the system's performance to changes in its parameters?
\end{itemize}

\vspace{-5pt}
\subsection{Experimental Setup}
%We implement HIGGS and make the source code publicly available.

{\bf Baselines.}
In the field of graph stream summarization supporting temporal range queries, SOTA methods include PGSS and Horae, with Horae excelling in query accuracy.
%In the realm of graph stream summarization supporting temporal range queries, existing methods include PGSS and Horae, with Horae being the state-of-the-art in terms of query precision.
However, their scalability is limited.
Therefore, we create stronger baselines by extending the SOTA graph stream summarization structure, i.e., Auxo, to support temporal range queries through the incorporation of Horae's range decomposition scheme, %{\color{red} called Auxo-cpt???}
yielding AuxoTime.
 %Therefore, we have extended the state-of-the-art scalable graph stream summarization, Auxo, by incorporating Horae's range decomposition scheme to facilitate temporal range queries.
Additionally, Horae has an optimized variant, Horae-cpt, aimed at reducing space overhead, which prompted us to also develop AuxoTime-cpt.
In the empirical study, we consider five competitive baselines, including PGSS, Horae, Horae-cpt, AuxoTime, and AuxoTime-cpt.

% We implement HIGGS and make the source code publicly available. In the evaluation, we employ state-of-the-art solutions, namely Horae and Auxo, as baselines. As Auxo does not inherently support temporal range queries, we leverage the range decomposition scheme proposed by Horae to extend Auxo and enable support for temporal range queries. Simultaneously, we also compared compacted Horae and compacted Auxo. This compacted data structure, proposed by Horae, serves as an optimization strategy primarily aimed at reducing space overhead.We utilize four large-scale real-world datasets, sorting each one chronologically to create graph streams.To flexibly support range queries at various temporal granularities, we set the temporal granularity to be consistent with the time unit of the dataset in use. This ensures that our approach, as well as the baselines, can not only accommodate large-grained range queries but also support range queries with the same small granularity as provided by the dataset.

\begin{table}[ht]
\centering
\vspace{-5pt}
\caption{Summary of Datasets}
\vspace{-5pt}
\label{tab:datasets}
\begin{tabular}{l|c|c|c|c}
\toprule
\bf Dataset & \bf Nodes & \bf Edges & \bf Time Span & \bf Time Slice \\
\midrule
\midrule
\bf Lkml & 63,399 & 1,096,440 & 2006-2013 & 1 second \\
\bf Wikipedia talk & 2,987,535 & 24,981,163 & 2001-2015 & 1 second \\
\bf Stackoverflow & 2,601,977 & 63,497,050 & 2009-2016 & 1 second \\
%\bf Wikipedia edits & 6,935,516 & 129,885,939 & - \\
\bottomrule
\end{tabular}
\end{table}
{\bf Datasets.}
We consider 3 real datasets commonly used for evaluating graph stream summarization, namely Lkml, Wiki-talk, and Stackoverflow, as shown in Table~\ref{tab:datasets}. %, and Wiki-edits.
%The specifics of the four large-scale real-world datasets utilized are detailed as follows.
The \textbf{Lkml} dataset \cite{kunegis2013konect} is a communication network of the Linux kernel mailing list, where nodes represent users (email addresses) and directed edges denote replies, encompassing 63,399 users and 1,096,440 replies. The \textbf{Wikipedia talk} (WT {\it in short}) \cite{kunegis2013konect} dataset captures the communication network of the English Wikipedia, with nodes representing users and edges indicating messages written on other users' talk pages, comprising 2,987,535 nodes and 24,981,163 messages. The \textbf{Stackoverflow} (SO {\it in short}) \cite{kunegis2013konect} dataset, collected from 2009 to 2016, represents an interaction network between users on Stack Overflow, featuring 2,601,977 nodes and 63,497,050 edges\footnote{The time slice for each dataset is set to 1 second to align with the time unit settings of each dataset.}.

{\bf Metrics.} The metrics primarily include \textit{average absolute error} (AAE), \textit{average relative error} (ARE), average query time, throughput, and space cost. AAE and ARE are commonly used for evaluating the accuracy of graph summarization, quantifying the error between the true values and the queried values. If there are \(p\) queries, each with a true value \(f_{i}\) and an estimated value \(\hat{f}_{i}\), then:
\begin{equation}
\small
  \mathrm{AAE} = \frac{1}{p} \sum_{i = 1}^{p} |f_{i} - \hat{f}_{i}| \text{ ~~~~~~~~}
  \mathrm{ARE} = \frac{1}{p} \sum_{i = 1}^{p} \frac{|f_{i} - \hat{f}_{i}|}{f_{i}}
\end{equation}
%Average query time primarily indicates the system's average response time across a series of query operations. Throughput measures the quantity of elements processed by the system per unit of time. Space overhead represents the size of the space resources required for storing or processing data in the system.
Average query time reflects the system’s performance on average query response for a series of query operations. Insertion throughput measures the number of elements processed per unit time, while insertion latency indicates the speed of processing insertions.
%%%%%%%%%%%%%%%%% zhao del 1025 %%%%%%%%%%%%%%%%%%%%%%%%%%%%%%%%
% Insertion throughput quantifies the number of elements processed by the system per unit of time.
Space overhead indicates the amount of main memory required to store or process data in the system.

{\bf Configuration.} Experiments were done on a machine equipped with a 32-core 2.6 GHz Xeon CPU, 384 GB RAM, and a 960 GB SSD. Baseline parameters were set following corresponding papers. For HIGGS, we set the number of optional addresses for each node to $4$ and $b$ to $3$. Therefore, for each edge, $4$ bits are required to store the index pair, indicating the position. Due to the close correlation between $Z$ and the degree of conflicts, we configured parameter $d_{1}$ to $16$ and $F_{1}$ to $19$ to ensure that the $Z$ value of HIGGS aligns with those of the baselines. For vertex and edge queries, we vary the query range length $L_q$ from $10^1$ to $10^7$. For each $L_q$, we randomly generated 100K edge queries and 10K vertex queries.
For path and subgraph queries, the path length is set to $[1,7]$ and subgraph size is set to $[50,350]$. For each path length and subgraph size, 1K queries are randomly generated. 
%We also randomly generated 1,000 path queries for each path length and 1,000 subgraph queries for each edge count, ensuring true weights were at least equal to their respective lengths or edge counts.
%In addition, for each dataset, we randomly generated 100,000 edge queries and 10,000 vertex queries for different range lengths $L_{q}$, with true weights of at least 1. We also randomly generated 1,000 path queries for each path length and 1,000 subgraph queries for each edge count, ensuring true weights were at least equal to their respective lengths or edge counts.}  
%\textcolor{blue}{Furthermore, we conducted deletion experiments by randomly sampling and deleting 5\% of the edges from each updated dataset, and evaluated the deletion throughput of HIGGS and other baselines.} \textcolor{red}{(R3O3)}
Each value reported is the average of 1,000 runs.
%%%%%%%%%%%%% zhao del 1025 %%%%%%%%%%%%%%%%%%%%%%%%%%%%%%%%%%
% The queries are randomly generated and each value reported is the average of 1,000 runs.

\subsection{Performance on Edge and Vertex Queries}
\label{subsec:queryacc}

\textit{\textbf{1) AAE vs. L}}: Fig.~\ref{fig:edgeQuery} (a--c) examine the accuracy of edge queries for all competitors in terms of AAE. In all experiments, the AAE of HIGGS is notably lower than those of other baselines, and this is consistent across all the three real graph stream datasets, i.e., Lkml, Wikipedia talk, and StackOverflow. For example, when the length of the query range $L_q$ equals $10^4$, the AAE of HIGGS is about $5$ orders of magnitude lower than that of the second-best competitor on Stackoverflow (Fig.~\ref{fig:edgeAaeSub3}).
Remarkably, the AAE of HIGGS is almost zero on the Lkml dataset (Fig.~\ref{fig:edgeAaeSub1}), implying 100\% query accuracy.
The superior performance of HIGGS is attributed to the elimination of conflicts of streaming items in the upper layers.
In contrast, other methodologies experience an accumulation of errors as the query range expands, resulting in inaccuracies caused by the error accumulation in the query processing.
Since the number of structure layers and the length of the query range are logarithmically related, the AAE of these methods increases logarithmically as the query range expands.
%{\color{red}However, in our approach, the cumulative error in our approach stems from the conflicts caused by the additional edges introduced when the range is extended. Therefore, even when the temporal range is extended to cover the entire graph stream, the conflicts in the HIGGS remain equivalent only to those encountered in a single-layer query in existing methods.}
However, our approach avoids the inter-layer error accumulation.
Even when the temporal range spans the entire graph stream, the conflicts in HIGGS are confined to the layer of leaf nodes.
This also explains why, when the query range is exceptionally large, the advantage of HIGGS over other competitors becomes smaller, since the task of TRQ degenerates into the task of count estimations of conventional (non-temporal) graph summaries like TCM. Although the accuracy of HIGGS is slightly lower, it remains orders of magnitude better than the accuracies of its competitors. Furthermore, the AAE of Horae-cpt and AuxoTime-cpt are higher than those of Horae and AuxoTime, respectively, as they decompose the temporal query range into more sub-ranges, leading to increased query conflicts and, consequently, a relative increase in AAE.

\textit{\textbf{2) ARE vs. L}}: Fig.~\ref{fig:edgeQuery} (d--f) depict the ARE results for edge queries across three datasets, as the query range length varies. Similar to the result observed in Fig.~\ref{fig:edgeQuery} (a--c), HIGGS consistently outperforms other baselines and maintains near-lossless performance
across all three datasets, showcasing its advantage. It is noteworthy that in Fig.~\ref{fig:edgeQuery} (d--f), the ARE of the baselines initially increase and then decrease.
This pattern is due to the relative error decreasing as the true edge weights increase with the lengthening of the query range.
\begin{figure*}[h]
  \vspace{-20pt}
  \centering
  \includegraphics[width=.7\textwidth, height=20pt]{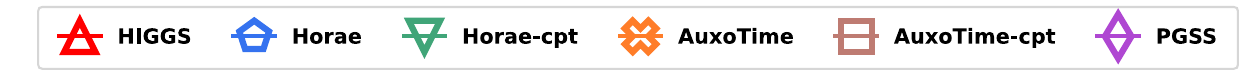}
  %\centering
  \begin{minipage}[b]{\columnwidth}
  \centering
      \input{exp-edge}
      \vspace{-14pt}
  \caption{Results on Edge Queries}
  \label{fig:edgeQuery}
  \end{minipage}
  \hfill
  \begin{minipage}[b]{\columnwidth}
  \centering
      \input{exp-node}
      \vspace{-14pt}
  \caption{Results on Vertex Queries}
  \label{fig:nodeQuery}
  \end{minipage}
  \centering
  \begin{minipage}[b]{\columnwidth}
  \centering
      \input{exp-path}
      \vspace{-14pt}
  \caption{Results on Path Queries}
  \label{fig:pathQuery}
  \end{minipage}
  \hfill
  \begin{minipage}[b]{\columnwidth}
  \centering
      \input{exp-subgraph}
      \vspace{-14pt}
  \caption{Results on Subgraph Queries}
  \label{fig:subgraphQuery}
  \end{minipage}
  \hfill
  \begin{minipage}[t]{\columnwidth}
  \centering
      \input{irSkew}
      \vspace{-12pt}
  \caption{Vertex Queries and Update Cost by Skewness}
  \label{fig:vqSkew}
  \end{minipage}
  \hfill
  \begin{minipage}[t]{\columnwidth}
  \centering
      \input{irVar}
      \vspace{-12pt}
  \caption{Vertex Queries and Update Cost by Variance}
  \label{fig:vqVar}
  \end{minipage}
  \vspace{-12pt}
\end{figure*}
%the ARE of baselines shows an initial increase followed by a decrease. This is attributed to the increase in true edge weights with the lengthening of the time range, resulting in a downward trend in relative error.

\textit{\textbf{3) Latency vs. L}}: Fig.~\ref{fig:edgeQuery} (g--i) illustrate the relationship between the query latency and the length of the query range, across all datasets.
% The HIGGS's query latency is on average two orders of magnitude lower than Horae's, three orders of magnitude lower than Horae-cpt's, one order of magnitude lower than Auxo's, and 1.5 orders of magnitude lower than Auxo-cpt's.
It shows that HIGGS and PGSS markedly outperform other baselines: when the query range length $L_{q}$
is set to $10^4$, HIGGS surpasses AuxoTime by an order of magnitude and improves on Horae by two orders of magnitude.
The superior performance of HIGGS stems primarily from its ability to access fewer summary layers given the same query range, and each layer accesses fewer nodes/buckets compared to AuxoTime and Horae.
PGSS exhibits competitive query latency (Fig.~\ref{fig:edgeQuery}(g--i)), but it falls short in query accuracy (Fig.~\ref{fig:edgeQuery}(a--f)).
Additionally, Horae's query efficiency is compromised by excessive accesses to the buffer structure, degrading its performance, as is the case with AuxoTime.
Moreover, the query latency of AuxoTime-cpt and Horae-cpt exceed those of AuxoTime and Horae, respectively, because they decompose the query range into more sub-ranges, increasing the query complexity from $O(\log L)$ to $O(\log^{2} L)$.

The performance of vertex queries is reported in Fig.~\ref{fig:nodeQuery}, where the trend is similar to that of edge queries. The details are thus omitted for brevity.

\begin{figure*}[ht]
    \centering
    \includegraphics[width=.7\textwidth, height=20pt]{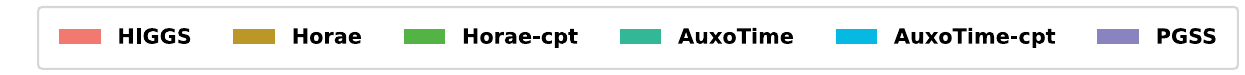}
    
    \begin{minipage}[t]{.33\columnwidth}
        \centering
        \includegraphics[width=\linewidth]{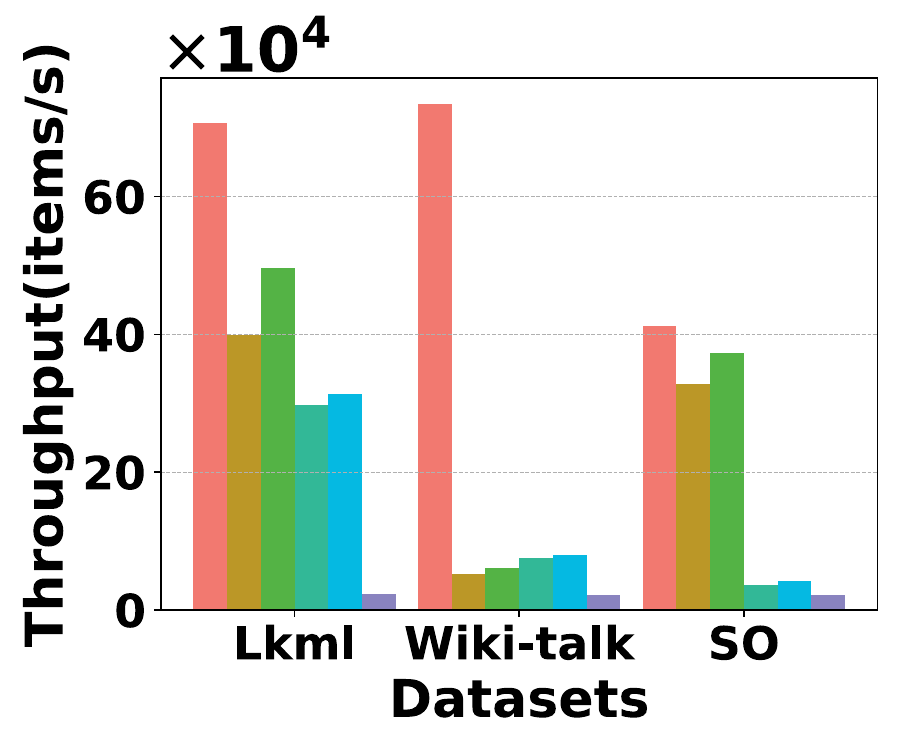}
        \vspace{-17pt}
        \caption{Insertion\\Throughput}
        \label{fig:insertThr}
    \end{minipage}
    \hfill
    \begin{minipage}[t]{.33\columnwidth}
        \centering
        \includegraphics[width=\linewidth]{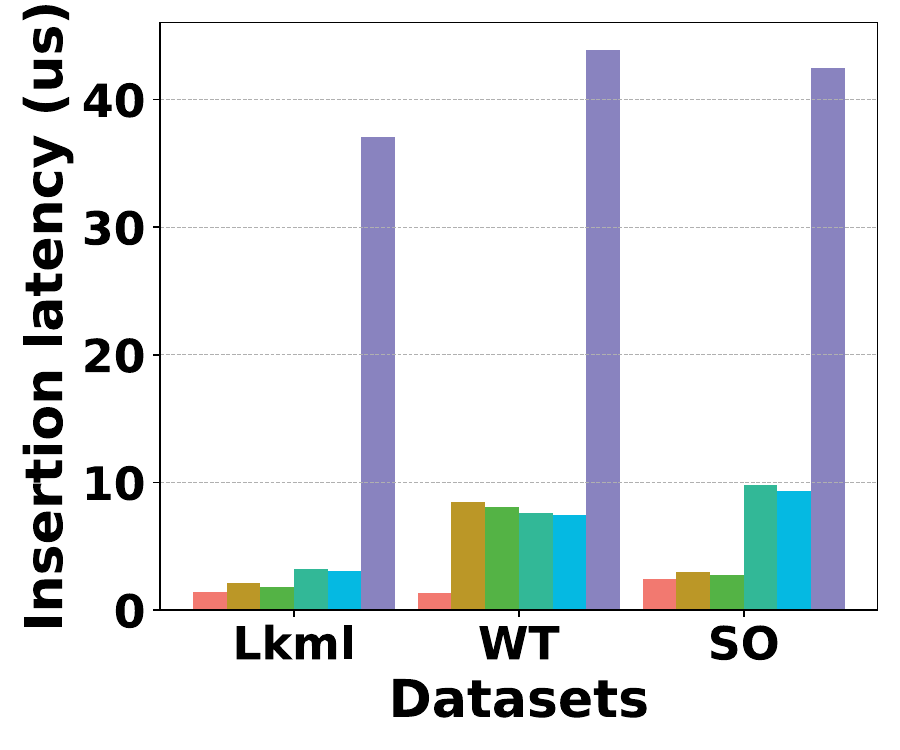}
        \vspace{-17pt}
        \caption{Insertion\\Latency}
        \label{fig:insertLat}
    \end{minipage}
    \hfill
    \begin{minipage}[t]{.33\columnwidth}
        \centering
        \includegraphics[width=\linewidth]{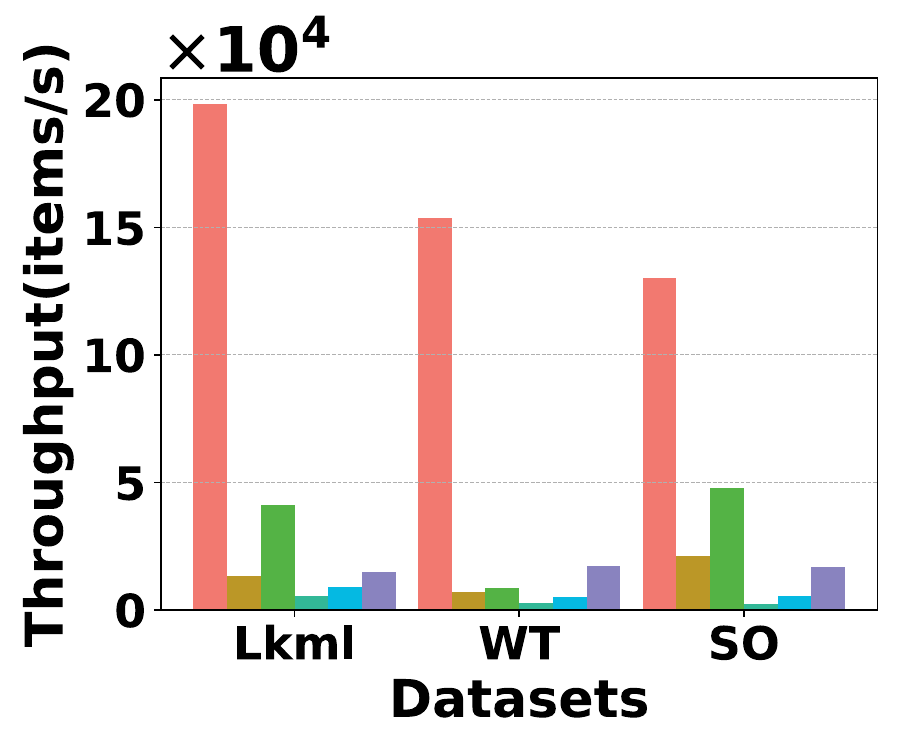}
        \vspace{-17pt}
        \caption{Deletion\\Throughput}
        \label{fig:deleteThr}
    \end{minipage}
    \hfill
    \begin{minipage}[t]{.33\columnwidth}
        \centering
        \includegraphics[width=\linewidth]{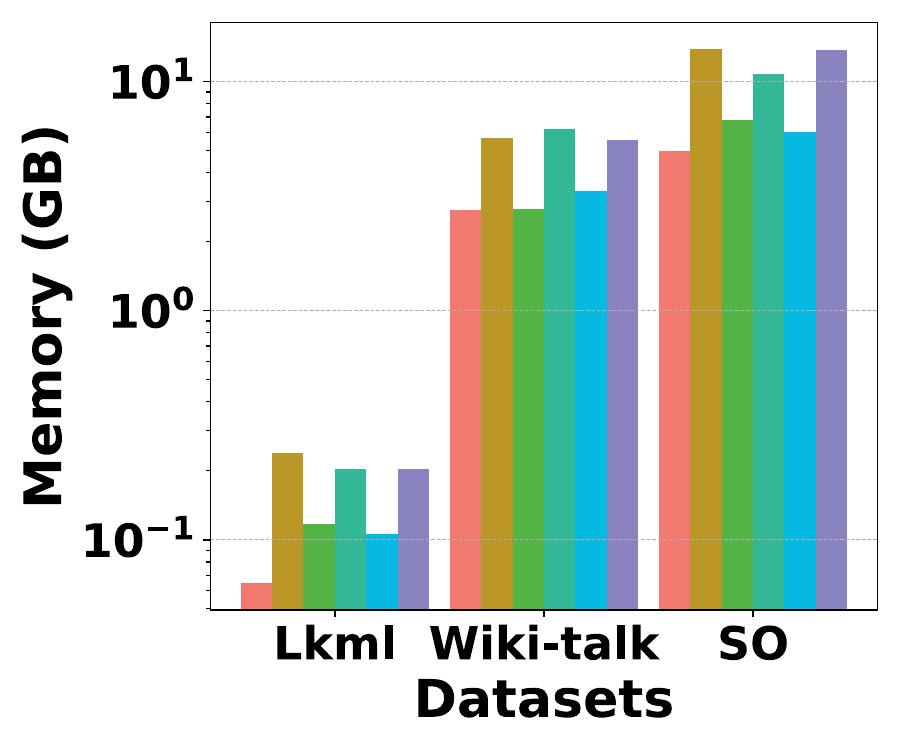}
        \vspace{-17pt}
        \caption{Space Cost}
        \label{fig:spaceCost}
    \end{minipage}
    \hfill
    \begin{minipage}[t]{.66\columnwidth}
      \begin{subfigure}{.5\textwidth}
        \centering
        \includegraphics[width=\linewidth]{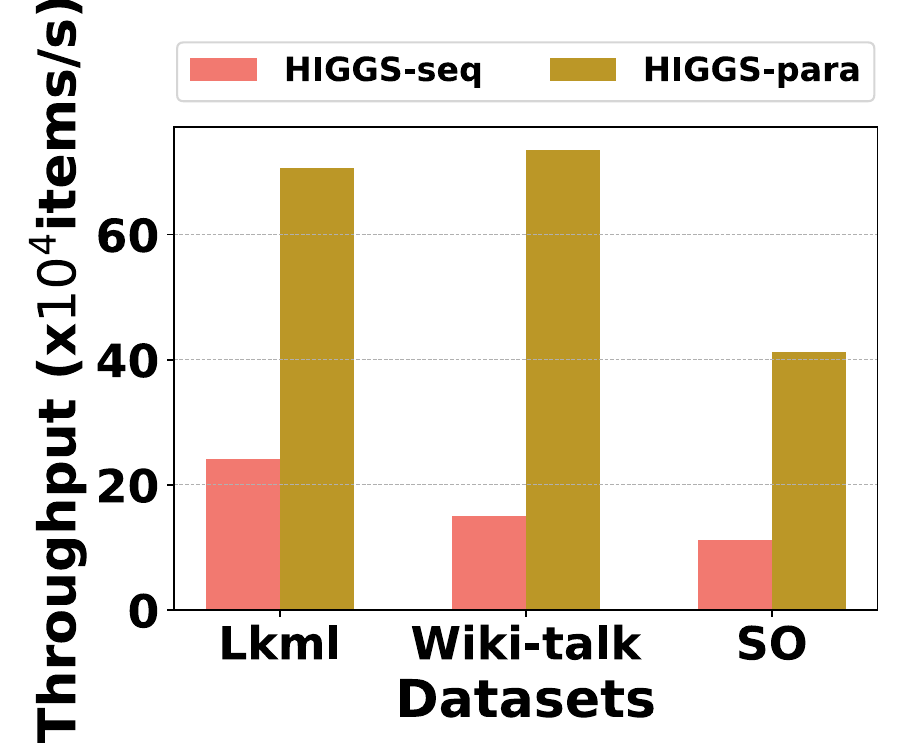}
        \vspace{-15pt}
        \caption{Throughput}
        \label{fig:optThroughput}
      \end{subfigure}\hfill
      \begin{subfigure}{.5\textwidth}
        \centering
        \includegraphics[width=\linewidth]{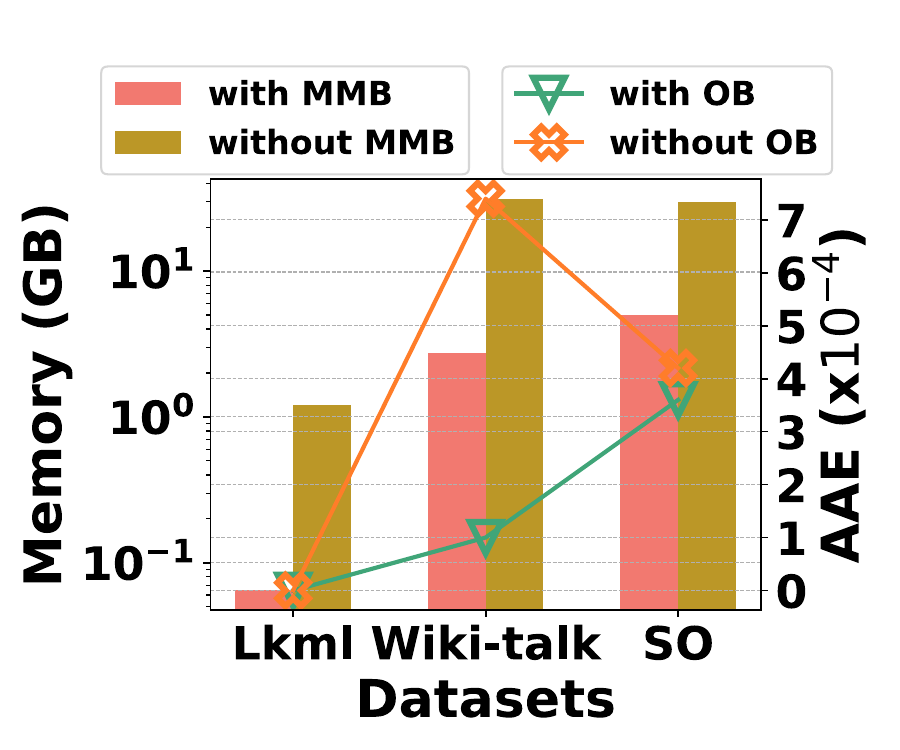}
        \vspace{-15pt}
        \caption{Space Cost \& AAE}
        \label{fig:optMemoryAae}
      \end{subfigure}\hfill
      \vspace{-6pt}
       \caption{Optimization}
       \label{fig:opt}
    \end{minipage}
    \vspace{-16pt}
\end{figure*}

\begin{figure}[ht]
    \begin{subfigure}{.48\columnwidth}
        \centering
      \includegraphics[width=\linewidth]{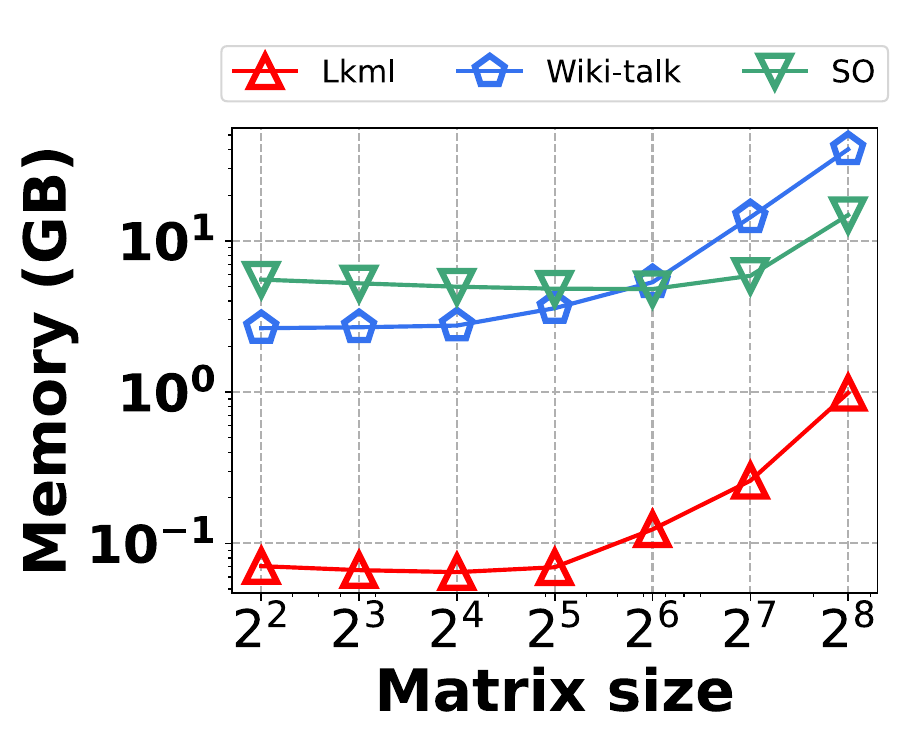}
      \vspace{-15pt}
      \caption{Space Cost vs. $d_{1}$}
      \label{fig:memoryHIGGS}
    \end{subfigure}\hfill
    \begin{subfigure}{.48\columnwidth}
        \centering
      \includegraphics[width=\linewidth]{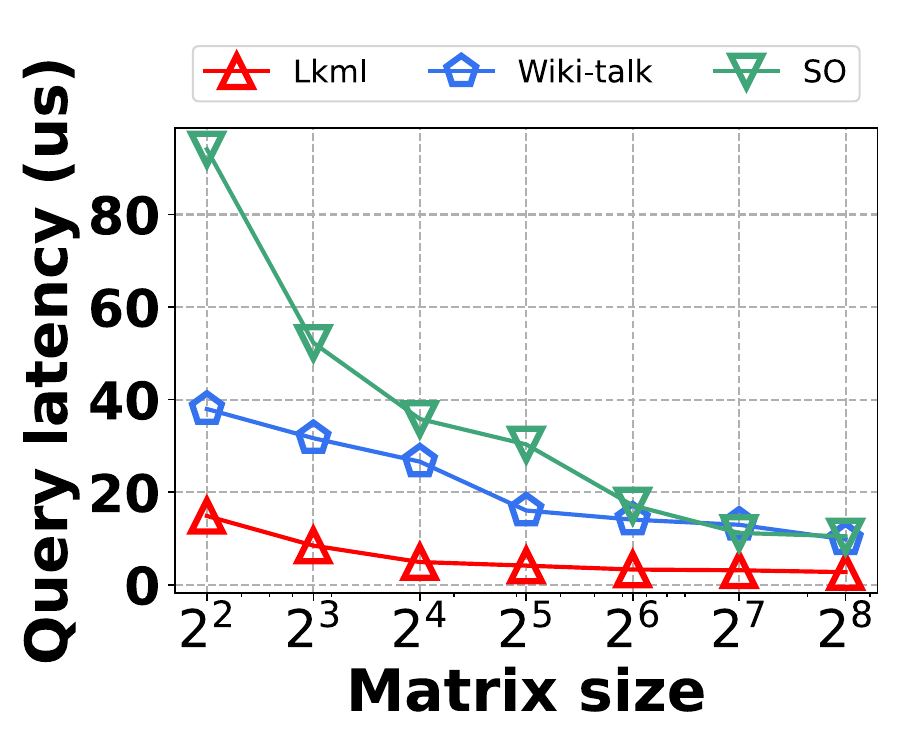}
      \vspace{-15pt}
      \caption{Query Latency vs. $d_{1}$}
      \label{fig:matrixSizeQuery}
    \end{subfigure}\hfill
    \vspace{-5pt}
     \caption{Parameter Analysis}
      \label{fig:matrixSizeQuery}
      \vspace{-18pt}
\end{figure}

\subsection{Performance on Path and Subgraph Queries}
\textit{\textbf{1) AAE vs. L}}: Fig.~\ref{fig:pathQuery} (a--c) report the performance on path queries, in particular, the relation between the accuracy (AAE) and the path length (number of hops), with the temporal range set to $10^5$. Compared to its competitors, HIGGS exhibits near-lossless performance across all three datasets, showing a clear advantage that grows with the number of hops in the path. This is due to HIGGS's cumulative precision benefits in edge queries, which are amplified in path queries since each path query consists of multiple edge queries. Additionally, Horae and AuxoTime demonstrate similar AAE, each outperforming their respective variants, Horae-cpt and AuxoTime-cpt. The absence of fingerprint utilization in PGSS for identifying stored edges leads to significant conflicts and a noticeable decrease in accuracy.
%\textit{\textbf{1) AAE vs. L}}: {\color{blue} Fig.~\ref{fig:pathQuery} (a-c) illustrates the relationship between the AAE for path queries and the number of hops in the path when the query range length is set to $10^5$, comparing the HIGGS and its competitors. The HIGGS demonstrates near-lossless performance across all three datasets, showcasing a pronounced advantage that expands as the number of hops in the path increases. This is attributed to the cumulative precision benefits in edge queries by the HIGGS, which are amplified in path queries since each path query comprises a series of edge queries. Moreover, Horae and AuxoTime exhibit nearly identical query precision, each outperforming their respective variants, Horae-cpt and AuxoTime-cpt. The lack of fingerprint utilization in PGSS for identifying stored edges results in significant conflicts, leading to a marked degradation in accuracy.}

\textit{\textbf{2) ARE vs. L}}: Fig.~\ref{fig:pathQuery} (d--f) show the trends of the values of ARE for path queries across three datasets as the number of hops increases, with the temporal range set to $10^{5}$. Consistent with the AAE trend, HIGGS maintains near 100\% query accuracy, significantly outperforming the competitors.

\textit{\textbf{3) Latency vs. L}}: Fig.~\ref{fig:pathQuery} (g--i) depict a comparative latency analysis of path queries among all the competitors and  datasets. HIGGS consistently leads, e.g., on Stackoverflow, it outperforms PGSS by 2.2x, AuxoTime by 30.8x, and Horae by two orders of magnitude when the path is 4 hops. The latency disadvantages of Horae-cpt and AuxoTime-cpt in edge querying are magnified here.

Fig.~\ref{fig:subgraphQuery} reports the performance on subgraph queries, including the accuracy (AAE\&ARE) and query latency. The results are quite similar to that on path queries, and are omitted due to page limits.

% showcase the changes in the AAE, ARE and latency for subgraph queries as a function of the size of subgraphs, following trends similar to those observed in path queries. Detailed discussion of these trends is omitted due to page limits.
% \vspace{-15pt}
\subsection{Performance on Graph Stream Irregularity} 
The irregularity in graph streams arises from skewed vertex degree distributions (skewness) and imbalanced edge arrivals (variance)\footnote{%To investigate skewness and variance individually, we isolate the influence of each.
We synthesize 6 datasets with varied skewness indicated by power-law exponents ranging from 1.5 to 3.0,
and 6 datasets with varied variances from 600 to 1,600.
Each dataset contains 100K nodes and 5M edges.}.
%each with a variance (calculated as a measure of edge arrival imbalance) of 1,000, as well as 6 datasets with variances from 600 to 1,600, each with a power-law exponent of 2.4.
Fig.~\ref{fig:vqSkew} presents the results of vertex queries and update cost under varied skewness. HIGGS significantly outperforms other baselines in AAE, achieving zero error at a skewness of 2.4, while the best alternative has an AAE of around 10. For query latency, HIGGS leads by approximately two orders of magnitude over the runner-up. Regarding space overhead and throughput at a skewness of 2.4, HIGGS exceeds Horae by 3.2 and 5.3 times, Horae-cpt by 1.6 and 7.8 times, Auxo by 3.1 and 2.9 times, Auxo-cpt by 1.6 and 2.5 times, and PGSS by 2.8 and 19.2 times. 
Fig.~\ref{fig:vqVar} shows the performance under different variances, with HIGGS still far ahead. Due to the page limitation, further details are omitted.

\subsection{Performance on Insertion Throughput and Latency}
Fig.~\ref{fig:insertThr} compares the performance of insertion throughput for different methods on different datasets.
%{\color{blue} Figure XX delineates the throughput comparisons of various technologies across different datasets.
HIGGS continues to hold a substantial lead, significantly outperforming all competitors. In particular, on Stackoverflow, it leads Horae, AuxoTime, Horae-cpt, AuxoTime-cpt, and PGSS by 1.3x, 11.2x, 1.1x, 9.7x, and 19.3x, respectively.
Additionally, Horae-cpt and AuxoTime-cpt achieve slightly higher throughput compared to Horae and AuxoTime,  because they store less data except the bottom layer, thus requiring fewer updates to stream item insertion. Fig.~\ref{fig:insertLat} presents a comparison of insertion latency among different methods. HIGGS outperforms Horae, Horae-cpt, AuxoTime, AuxoTime-cpt, and PGSS on Stackoverflow by 1.2x, 1.1x, 4.0x, 3.8x, and 17.4x, respectively.
%Additionally, Horae-cpt and AuxoTime-cpt, due to storing less data beyond the first layer and thus requiring fewer updates to elements, achieve slightly higher throughput compared to Horae and AuxoTime.

\subsection{Performance on Deletion Throughput}
Fig.~\ref{fig:deleteThr} illustrates the deletion throughput performance across three datasets, where HIGGS consistently outperforms other baselines. On Stack Overflow, HIGGS is faster than Horae by 6.1x, Horae-cpt by 2.7x, AuxoTime by 56.0x, AuxoTime-cpt by 24.3x, and PGSS by 7.7x.

\subsection{Performance on Space Cost} Fig.~\ref{fig:spaceCost} presents the space overhead for different methods across different datasets.
It shows that HIGGS consistently achieves the lowest space overhead on all datasets. Specifically, on Stackoverflow, it achieves a space savings of 26.8\% compared to Horae-cpt, 17.5\% compared to AuxoTime-cpt, 64.1\% over Horae, 53.9\% against AuxoTime, and 63.6\% compared to PGSS.
This implies that HIGGS uses less space to achieve markedly better query performance.

\subsection{Optimization}
%{\color{blue}
Fig.~\ref{fig:opt} evaluates the effectiveness of proposed optimizations of parallelization, multiple mapping buckets (MMB), and overflow buckets (OB).
%%%%%%%%%%%% zhao del 1028 %%%%%%%%%%%%%%%%%%%%%%%%%%%%
% The target of parallelization is to accelerate the processing of stream item insertion.
Fig.~\ref{fig:opt}
(a) shows the throughput of HIGGS with and without the parallelization optimization across three datasets, where HIGGS's throughput has increased at least 3x by adopting parallelization.
%%%%%%%%%%%%%%% zhao del 1028 %%%%%%%%%%%%%%%%%%%%%%%%%%%%
% The purpose of MMB is to enhance the query accuracy.
In Fig.~\ref{fig:opt} (b), it shows that employing MMB substantially improves space efficiency and incorporating OB improves accuracy. For example, on StackOverflow, the MMB mechanism brings in about 6x increase in space efficiency compared to its absence, while the OB mechanism leads to a 14.3\% improvement in accuracy. %By default, the experiment utilizes both optimizations.

%the with MMB and 'without MMB' represent scenarios with and without the use of multi-bucket mapping, respectively, while 'with OB' and 'without OB' denote the presence or absence of overflow blocks. The results demonstrate that employing multi-bucket mapping substantially improves space efficiency, and incorporating overflow blocks enhances accuracy. On StackOverflow, the application of multi-bucket mapping achieved a 6x increase in performance compared to its absence; the use of overflow blocks led to a 14.3\% improvement in accuracy. By default, the experiment utilizes both optimizations.}

\subsection{Parameters}
Fig.~\ref{fig:matrixSizeQuery} investigates the impact of parameters about matrix size at leaf nodes on space overhead and query latency.
The results demonstrate that typically, larger matrices at the leaf nodes lead to higher space cost but reduced query latency. For instance, with a matrix size of 16, the space overhead on Stackoverflow is 5,085MB, and the query latency is 35.8 $\mu$s, which is significantly outperforming other competitors. Therefore, we recommend setting $d_{1}$ to 16 to achieve a good balance between the space overhead and query latency.

%%%%%%%%%%%% zhao del 0627 %%%%%%%%%%%%%%%%
% {\color{blue} To investigate the impact of matrix size at leaf nodes on space overhead and query latency, experiments were conducted across three datasets, with the results depicted in Figures XX and XX. It is evident that as the matrix size increases from 4 to 32, there is no significant increase in space overhead across all datasets, though it begins to rise thereafter. Concurrently, query latency decreases as the matrix size increases. Typically, setting the matrix size to 16 or 32 yields more stable performance across various datasets.}

\section{Conclusion}
We propose HIGGS, a novel item-based, bottom-up hierarchical structure designed for summarizing graph streams with temporal information. HIGGS leverages its hierarchical design to localize storage and query processing, effectively confining changes and hash conflicts within smaller, manageable sub-trees. Compared to existing approaches, HIGGS achieves a notable enhancement in overall performance, supported by both a theoretical analysis and empirical studies. Extensive experiments on real graph streams show that HIGGS can improve accuracy by more than $3$ orders of magnitude, reduce space overhead by an average of $30$\%, increase throughput by $5$+ times, and reduce query response time by nearly $2$ orders of magnitude. In future research, it is of interest to extend HIGGS to support advanced graph stream variants, including foundational variants like heterogeneous graph streams with diverse node and edge types, as well as application-specific variants such as spatiotemporal, web, and social graph streams.

\section*{Acknowledgments}
This work was supported by the NSFC under Grants 62472400 and 62072428. Xike Xie is the corresponding author.
%%%%%%%%%%%%%%%%%%%%%% zhao del 10.9 %%%%%%%%%%%%%%%%%%%%%%%%%%%%%%%%%%%
% In the future, we plan to explore extending the HIGGS to support more sophisticated variants of graph streams. 

% \textcolor{blue}{In the future, we plan to explore extending the HIGGS to support more sophisticated variants of graph streams, including heterogeneous graph streams with different types of nodes and edges, as well as graph streams incorporating temporal windows defined by data scale or time span to focus on more time-sensitive data.} \textcolor{red}{(R3O6)}

\newpage
%\clearpage

\bibliographystyle{IEEEtran}
\bibliography{reference.bib}
\end{document}